\documentclass{article}

\usepackage{graphicx, color}
\usepackage{amsmath} 
\usepackage{amssymb}
\usepackage{amsthm}
\usepackage{amsxtra}
\usepackage{amscd}
\usepackage{bbm}
\usepackage{mathrsfs}
\usepackage{bm}

\renewcommand{\b}[1]{\boldsymbol{\mathrm{#1}}} 
\newcommand{\bb}{\mathbb} 
\renewcommand{\cal}{\mathcal} 
\newcommand{\scr}{\mathscr} 
\newcommand{\fra}{\mathfrak} 
\newcommand{\ol}[1]{\overline{#1} \!\,} 

\newcommand{\wt}{\widetilde}

\newcommand{\ee}{\mathrm{e}} 
\newcommand{\ii}{\mathrm{i}} 
\newcommand{\dd}{\mathrm{d}}

\newcommand{\deq}{\mathrel{\mathop:}=}
\newcommand{\eqd}{=\mathrel{\mathop:}}
\newcommand{\umat}{\mathbbmss{1}} 
\renewcommand{\epsilon}{\varepsilon}
\renewcommand{\leq}{\leqslant}
\renewcommand{\geq}{\geqslant}

\newcommand{\ind}[1]{\b 1 (#1)}
\newcommand{\indb}[1]{\b 1 \pb{#1}}

\renewcommand{\P}{\mathbb{P}}
\newcommand{\E}{\mathbb{E}}
\newcommand{\R}{\mathbb{R}}

\newcommand{\N}{\mathbb{N}}
\newcommand{\Z}{\mathbb{Z}}

\newcommand{\p}[1]{({#1})}
\newcommand{\pb}[1]{\bigl({#1}\bigr)}
\newcommand{\pB}[1]{\Bigl({#1}\Bigr)}
\newcommand{\pbb}[1]{\biggl({#1}\biggr)}
\newcommand{\pBB}[1]{\Biggl({#1}\Biggr)}

\newcommand{\qb}[1]{\bigl[{#1}\bigr]}

\newcommand{\qbb}[1]{\biggl[{#1}\biggr]}
\newcommand{\qBB}[1]{\Biggl[{#1}\Biggr]}

\newcommand{\hb}[1]{\bigl\{{#1}\bigr\}}

\newcommand{\hbb}[1]{\biggl\{{#1}\biggr\}}

\newcommand{\abs}[1]{\lvert #1 \rvert}
\newcommand{\absb}[1]{\big\lvert #1 \big\rvert}

\newcommand{\absbb}[1]{\bigg\lvert #1 \bigg\rvert}

\newcommand{\norm}[1]{\lVert #1 \rVert}
\newcommand{\normb}[1]{\big\lVert #1 \big\rVert}

\newcommand{\normbb}[1]{\bigg\lVert #1 \bigg\rVert}

\newcommand{\scalar}[2]{\langle{#1} \mspace{2mu}, {#2}\rangle}
\newcommand{\scalarb}[2]{\big\langle{#1} \mspace{2mu}, {#2}\big\rangle}

\newcommand{\bra}[1]{\langle #1 |}

\newcommand{\ket}[1]{| #1 \rangle}

\DeclareMathOperator{\tr}{Tr}

\DeclareMathOperator{\sgn}{sgn}

\theoremstyle{plain} 
\newtheorem{theorem}{Theorem}[section]
\newtheorem*{theorem*}{Theorem}
\newtheorem{lemma}[theorem]{Lemma}
\newtheorem*{lemma*}{Lemma}
\newtheorem{corollary}[theorem]{Corollary}
\newtheorem*{corollary*}{Corollary}
\newtheorem{proposition}[theorem]{Proposition}
\newtheorem*{proposition*}{Proposition}

\newtheorem*{definition*}{Definition}

\newtheorem*{example*}{Example}
\newtheorem{remark}[theorem]{Remark}

\newtheorem*{remark*}{Remark}
\newtheorem*{remarks*}{Remarks}

\usepackage{cite}

\usepackage{graphicx}
\usepackage[nottoc,notlof,notlot]{tocbibind}

\usepackage{psfrag}

\numberwithin{equation}{section}
\numberwithin{figure}{section}

\begin{document}
\author{
L\'aszl\'o Erd\H os${}^1$\thanks{Partially supported
by SFB-TR 12 Grant of the German Research Council.}, Antti Knowles${}^2$\thanks{Partially supported by U.S.\ National 
Science Foundation Grant DMS 08--04279.} \\ \\
Institute of Mathematics, University of Munich, \\
Theresienstr.\ 39, D-80333 Munich, Germany \\ lerdos@math.lmu.de ${}^1$ \\ \\
Department of Mathematics, Harvard University\\
Cambridge MA 02138, USA \\  knowles@math.harvard.edu ${}^2$ \\ \\
\\}
\title{Quantum Diffusion and Eigenfunction
Delocalization in a Random Band Matrix Model}
\date{20 September 2010}
\maketitle

\begin{abstract}
We consider Hermitian and symmetric random band matrices $H$ in $d \geq 1$ dimensions.
The matrix elements $H_{xy}$, indexed by $x,y \in \Lambda \subset \Z^d$, are
independent, uniformly distributed random variables if $\abs{x-y}$ is less than the band width $W$,
and zero otherwise.
We prove that the time evolution of a quantum particle
subject to the Hamiltonian $H$ is diffusive on time scales $t\ll W^{d/3}$. We also show that the localization length of
the eigenvectors of $H$ is larger than a factor $W^{d/6}$ times the band width.
All results are uniform in the size $\abs{\Lambda}$ of the matrix.
\end{abstract}

\vspace{1cm}

\medskip



\section{Introduction}

The general formulation of the universality conjecture for disordered systems
states that there are two distinctive  regimes depending on the energy and
the disorder strength. In the strong disorder regime, the eigenfunctions
are localized and the local spectral statistics are Poisson.
In the weak disorder regime, the eigenfunctions are delocalized
and the local statistics coincide with those of a Gaussian 
matrix ensemble.

Random band matrices are natural intermediate models
to study eigenvalue statistics and quantum propagation in disordered systems
as they interpolate between  Wigner matrices
and random Schr\"odinger operators. Wigner matrix ensembles
represent mean-field models without spatial structure, where the quantum transition
rates  between any
two sites are i.i.d.\ random variables with zero expectation. In the celebrated Anderson model
\cite{And}, only a random on-site potential $V$ is present in addition to  a short range 
deterministic hopping (Laplacian) on a graph that is typically a regular box in $\Z^d$.

For the Anderson model, a fundamental open question is to establish the
metal-insulator transition, i.e.\ to
show that in $d\geq 3$ dimensions the eigenfunctions of  $-\Delta+\lambda V$
are delocalized for small disorder $\lambda$.
 The localization regime at large disorder or near the spectral
edges has been well  understood by Fr\"ohlich and Spencer
with the multiscale technique \cite{FS, FMSS}, and later by Aizenman and Molchanov
by the  fractional moment method \cite{AiMo}; many other
works have since contributed to this field. In particular, it
has been established that the local eigenvalue statistics are Poisson \cite{Mi} and
that the eigenfunctions are exponentially localized
with an upper bound on the localization length that diverges
 as the energy parameter approaches the presumed phase transition 
point \cite{Spen, Elg}.

The progress in the delocalization regime has been much slower. For the Bethe lattice, 
corresponding to the infinite-dimensional case, delocalization has been established in 
\cite{Kl, ASW, FHS}.
In finite dimensions only partial results are available. The existence of an absolutely 
continuous spectrum (i.e.\ extended states) has been shown for a rapidly decaying potential, 
corresponding to a scattering regime \cite{RS, B, D}.  Diffusion
has been established  for a heavy quantum particle immersed in a
phonon field in $d\geq 4$ dimensions \cite{FdeR}.
For the original Anderson Hamiltonian  with a small coupling constant $\lambda$,
the eigenfunctions have a
localization length of at least $\lambda^{-2}$ (see \cite{Ch}). The time and space scale 
$\lambda^{-2}$ corresponds to
the kinetic regime where the quantum evolution can be
modelled by a linear Boltzmann equation \cite{Sp1, EY}. Beyond this
time scale the dynamics is diffusive. This has been established
in the scaling limit $\lambda\to0$
up to time scales $t\sim \lambda^{-2-\kappa}$ with an explicit $\kappa>0$
in \cite{ESY1, ESY2, ESY3}.
There are no rigorous results on the local spectral statistics of the Anderson model, but
it is conjectured -- and supported by numerous arguments in
the physics literature, especially by supersymmetric methods (see \cite{Efe}) -- that the local 
correlation function
of the eigenvalues of the finite volume Anderson
model follows the GOE statistics in the thermodynamic limit.

Due to their mean-field character, Wigner matrices are simpler to study
than the Anderson model  and  they are always in the delocalization regime.
The complete delocalization of the eigenvectors was proved in \cite{ESY4}. The local
spectral statistics in the bulk are universal, i.e.\ they follow the statistics
of the corresponding Gaussian ensemble (GOE, GUE, GSE), depending
on the symmetry type of the matrix (see \cite{M} for 
explicit formulas).  For an arbitrary single entry
distribution, bulk universality has been proved
recently in \cite{EPRSY, ESY5, ESYY} for all symmetry classes.
A different proof was given  in \cite{TV} for the Hermitian case.

Random band matrices $H=\{H_{xy}\}_{x,y\in \Gamma}$ represent  systems on a large finite graph 
$\Gamma$ with a metric. The matrix elements between two sites, $x$ and $y$, are independent 
random variables with a variance $\sigma_{xy}^2 \deq \E |H_{xy}|^2$ depending on the distance 
between the two sites. The variance typically decays with the distance on a characteristic 
length scale $W$, called the \emph{band width} of $H$. This terminology comes from the simplest 
one-dimensional
model where the graph is a path on $N$ vertices, labelled
by $\Gamma=\{1,2,\ldots, N\}$, and the matrix elements $H_{xy}$
vanish if $|x-y|\geq W$. If $W=N$ and all variances
are equal, we recover the usual Wigner matrix. The case $W=O(1)$
is a one-dimensional Anderson-type model with random hoppings at bounded range.  
Higher-dimensional
models are obtained if the graph $\Gamma$
is a box in $\Z^d$. For more general
random band matrices and for a systematic presentation,
see \cite{Spe}.

Since the one-dimensional Anderson-type models are always in the localization regime,
varying the band width $W$ offers a possibility to
test the localization-delocalization transition between an
Anderson-type model and the Wigner ensemble.
Numerical simulations and theoretical arguments based on supersymmetric methods \cite{Fy} suggest that the local 
eigenvalue statistics change from Poisson, for $W\ll N^{1/2}$, to
GOE (or GUE), for $W\gg N^{1/2}$. The eigenvectors
are expected to have a localization length $\ell$ of order $W^2$.
In particular the eigenvectors are fully delocalized for $W\gg N^{1/2}$.
In two dimensions the localization length is expected to be
exponentially large in $W$; see \cite{Abr}. In accordance with
the extended states conjecture for the Anderson model,
the localization length is expected to be macroscopic, $\ell \sim N$,
 independently of the band width in $d \geq 3$ dimensions.

Extending the techniques of the rigorous proofs for Anderson localization, Schenker has recently
proved   the upper bound $\ell\leq W^8$ for the localization
length in $d=1$ dimensions \cite{Sch}.   In this paper we prove a counterpart
of this result from the side of delocalization. More precisely,
we show a lower bound $\ell\geq W^{1+d/6}$ for the
eigenvectors of $d$-dimensional band matrices
with uniformly distributed entries.  We remark that the lower bound $\ell\geq W$ was proved 
recently in \cite{EYY}
for very general band matrices.

On the spectral side, we mention that,
 apart from the semicircle law (see \cite{AZ , gui, EYY} for $d=1$ and \cite{DPS} for $d=3$),
the question of bulk universality of {\it local}
 spectral statistics for band matrices is mathematically open
even for $d=1$. In the spirit of the general conjecture,  one expects GUE/GOE statistics in the bulk for the 
delocalization
regime, $W\gg N^{1/2}$. The GUE/GOE statistics have recently been established \cite{EYY} for a class of generalized
Wigner matrices, where the variances of different matrix elements are not necessarily identical, but are of comparable 
size, i.e.\ $\E |H_{xy}|^2 \sim \E |H_{x'y'}|^2$;
in particular, the band width is still macroscopic ($W\sim N$).

Supersymmetric methods offer a very attractive approach to study 
the delocalization transition in band matrices 
but the rigorous control of the functional integrals away from the
saddle points is difficult and it
has been performed only for the density of states \cite{DPS}.
Effective models that emerge near the saddle points can be
more accessible to rigorous mathematics. 
Recently Disertori, Spencer and Zirnbauer studied a related statistical
mechanics model that
is expected to reflect the Anderson localization and delocalization 
transition for real symmetric band matrices.
They proved a quasi-diffusive
estimate for the two-point correlation functions 
in a  three dimensional supersymmetric hyperbolic nonlinear sigma model
at low temperatures \cite{DSZ}. Localization was also established in the 
same model at high temperatures \cite{DS}.

We also mention that band matrices are not the only possible interpolating models
to mimic the metal-insulator transition. Other examples include
the Anderson model with a spatially 
decaying potential \cite{B, KKO}
and a quasi one-dimensional model with a weak on-site potential
for which a transition in the sense of local spectral statistics has been established in
\cite{BdeR, VV}.

A natural approach to study the delocalization regime is to show
that the quantum time evolution is diffusive on large scales.
We normalize the matrix entries so that the rate of quantum jumps
is of order one. The typical distance of a single jump is
the band width $W$. If the jumps were independent, the typical distance travelled in time $t$ 
would be $W\sqrt{t}$.
Using the argument of \cite{Ch}, we show
that a typical localization length $\ell$ is incompatible with
a diffusion on spatial scales larger than $\ell$. Thus we obtain
$\ell \geq W\sqrt{t}$, provided that the diffusion
approximation can be justified up to time $t$.

The main result of this paper is that
the quantum dynamics of the $d$-dimensional band matrix is given by
 a superposition of heat kernels up to time scales $t \ll W^{d/3}$.
Although diffusion is expected to hold up to time $t\sim W^2$
for $d=1$ and up to any time for $d\geq 3$ (assuming the
thermodynamic limit has been taken), our method can follow
the quantum dynamics only up to $t\ll W^{d/3}$.  The threshold exponent $d/3$ originates in technical estimates on 
certain Feynman graphs; going beyond the exponent $d/3$ would
require a further resummation of certain four-legged subdiagrams (see Section~\ref{sec:further}).

Finally, we remark that our method also yields a bound on the largest eigenvalue of a band matrix; see Theorem 3.4 in 
the forthcoming paper \cite{erdosknowles} for details.


\subsubsection*{Acknowledgements}
The problem of diffusion for random band matrices originated
from several discussions with H.T.\ Yau and J.\ Yin. The authors are especially
 grateful to J.\ Yin for various insights 
and for pointing out an improvement in the counting of the skeleton diagrams.

\section{The Setup}
Let the dimension $d \geq 1$ be fixed and consider the $d$-dimensional lattice $\Z^d$
equipped with  the Euclidean norm $\abs{\cdot}_{\Z^d}$ 
(any other norm would also do). We index points of $\Z^d$ with $x,y,z,\dots$. Let $W > 1$ 
denote a large parameter (the \emph{band width}) and define
\begin{equation*}
M \;\equiv\; M(W) \;\deq\; \absb{\{x \in \Z^d \,:\, 1 \leq \abs{x}_{\Z^d} \leq W \}}\,,
\end{equation*}
the number of points at distance at most $W$ from the origin.
In the following we tacitly make use of the obvious relation
$M \sim\; C W^d$. For notational convenience, we use both $W$ and $M$ in the following.

In order to avoid dealing with the infinite lattice directly,
we restrict the problem to a finite periodic lattice $\Lambda_N$ of linear size $N$.
More precisely,
for $N \in \N$, we set
\begin{equation*}
\Lambda_N \;\deq\; \{-[N/2], \dots, N - 1 - [N/2]\}^d \;\subset\; \Z^d\,,
\end{equation*}
a cube with side length $N$ centred around the origin. Here $[\cdot]$ denotes integer part.
We regard $\Lambda_N$ as periodic, i.e.\ we equip it with periodic addition and the periodic 
distance
\begin{equation*}
\abs{x} \;\deq\; \inf \{\abs{x + N \nu}_{\Z^d} \,:\, \nu \in \Z^d\}\,.
\end{equation*}
Unless otherwise stated, all summations $\sum_x$ are understood to mean $\sum_{x \in 
\Lambda_N}$.

We consider random matrices $H^\omega \equiv H$ whose entries $H_{xy}$ are indexed by $x,y \in 
\Lambda_N$. Here $\omega$ denotes the running element in probability space. The large parameter 
of the model is the band width $W$. We shall always assume that
$N \geq W M^{1/6}$. Under this condition all our results hold uniformly in $N$.

We assume that $H$ is either Hermitian or symmetric.
The entries $H_{xy}$ satisfying $1 \leq \abs{x - y} \leq W$ are i.i.d.\ (with the obvious 
restriction that $H_{yx} = \ol{H_{xy}}$). In the Hermitian case they are uniformly distributed 
on a circle of appropriate radius in the complex plane,
\begin{subequations} \label{hhormalization}
\begin{equation}\label{hnormalization_1}
H_{xy} \;\sim\; \frac{1}{\sqrt{M - 1}} \, \mathrm{Unif}(\bb S^1)\,, \qquad 1 \leq \abs{x - y} \leq W\,.
\end{equation}
In the symmetric case they are Bernoulli random variables,
\begin{equation} \label{hnormalization_2}
\P\pbb{H_{xy} = \frac{1}{\sqrt{M - 1}}} \;=\; \P\pbb{H_{xy} = \frac{-1}{\sqrt{M - 1}}} \;=\; \frac{1}{2}\,, \qquad 1 
\leq \abs{x - y} \leq W\,.
\end{equation}
\end{subequations}
If $\abs{x - y} \notin [1, W]$ then $H_{xy} = 0$. An important 
consequence of our assumptions \eqref{hnormalization_1} 
and \eqref{hnormalization_2} is
\begin{equation} \label{deterministic norm}
\abs{H_{xy}}^2 \;=\; \frac{1}{M-1} \, \ind{1 \leq \abs{x - y} \leq W}\,.
\end{equation}

We remark that the assumption that the matrix entries have the special
 form \eqref{hnormalization_1} or 
\eqref{hnormalization_2} is not necessary for our results to hold. 
We make it here because it greatly simplifies our 
proof. The reason for this is that, as observed by Feldheim and Sodin \cite{FSo,So1},
 the condition \eqref{deterministic 
norm} allows one to obtain a simple algebraic expression for the nonbacktracking 
powers of $H$; see Lemma \ref{lemma 
recursion for nonbacktracking random walks}.

In the forthcoming paper \cite{erdosknowles} we extend our results
to random matrix ensembles in which the matrix elements $H_{xy}$ are 
allowed to have a general distribution (and thus in 
particular a genuinely random absolute value); moreover their 
variances $\E \abs{H_{xy}}^2$ are given by a general 
profile
on the scale $W$ in $x-y$ (as opposed to the step function profile in 
\eqref{deterministic norm}). Under these 
assumptions, the algebraic identity of 
Lemma~\ref{lemma recursion for nonbacktracking random walks} is no longer exact, 
and needs to be amended with additional random terms. 
The resulting graphical expansion is considerably more involved 
than in the case \eqref{deterministic norm}, and its control 
requires essential new ideas. However, the fundamental 
mechanism underlying quantum diffusion for band matrices is 
already apparent in the special case \eqref{deterministic 
norm} discussed in this paper.

Let $\alpha \in \fra A \deq \{1, \dots, \abs{\Lambda_N}\}$ index the orthonormal basis 
$\{\psi_\alpha^\omega\}_{\alpha \in \fra A}$ of eigenvectors of the matrix $H^\omega$, i.e.
$H^\omega \psi_\alpha^\omega = \lambda^\omega_\alpha \psi_\alpha^\omega$,
where $\lambda_\alpha^\omega \in \R$.
The normalization of the matrix elements is chosen
 in such a way that the typical
eigenvalue of the matrix is of order one:
\begin{equation*}
 \frac{1}{\abs{\fra A}}\sum_\alpha \E \lambda_\alpha^2 \;=\;
  \frac{1}{\abs{\fra A}} \E \tr H^2 \;=\; \frac{1}{\abs{\fra A}} \E \sum_{x,y} |H_{xy}|^2 \;=\; 
\frac{M}{M - 1}\,.
\end{equation*}

\section{Scaling and results}
The central quantity of our analysis is
\begin{equation*}
\varrho(t,x) \;\deq\; \E \, \absb{\scalarb{\delta_x}{\ee^{-\ii t H / 2} \delta_0}}^2\,,
\end{equation*}
where $\delta_x\in \ell^2(\Lambda_N)$ denotes the standard basis vector, defined by $(\delta_x)_y = \delta_{xy}$.  The 
factor $1/2$ is a convenient normalization
 since, by a standard result of random 
matrix theory, the spectrum of $H/2$ is asymptotically equal to the unit interval $[-1,1]$.
The function $\varrho(t,x)$ describes the ensemble average of the
 quantum transition probability of a
particle starting from position $0$ ending up at position $x$ after time $t$. Note that $\sum_x 
\varrho(t,x) =1$ for any $t\in \R$.
Heuristically, the particle performs a series of random jumps of size $W$. The typical number 
of jumps in time $t = O(1)$ is of order one.
 Indeed, by first order perturbation theory,
the small-times probability distribution for $ 1 \leq \abs{x} \leq W$ is given by
\begin{equation*}
\varrho(t,x) \;\sim\; \E \, \absb{\scalarb{\delta_x}{(\umat - \ii t H/2) \delta_0}}^2 \;=\; 
\frac{t^2}{4} \, \E \, \abs{H_{x0}}^2 \;=\; \frac{t^2}{4} \, \frac{1}{M - 1}\,,
\end{equation*}
up to higher order terms in $t$. Thus $\sum_{x\ne0} \varrho(t,x)$ is an $O(1)$ quantity, separated away 
from zero, indicating that the distance from the origin
is of $O(W)$ for times $t\sim O(1)$.

In time $t$ the particle performs $O(t)$ jumps of size $O(W)$. We
 expect that the jumps are approximately independent 
and the trajectory is
a random walk consisting of $O(t)$ steps with size $O(W)$ each. Thus,
the typical distance from the origin is of order $t^{1/2} W$. We
 rescale time and space $(t,x) \mapsto (T,X)$ so as to make the macroscopic quantities $T$ and $X$ of order one, i.e.\ 
we set
\begin{equation*}
t \;=\; \eta T \,, \qquad x \;=\; \eta^{1/2} W X\,,
\end{equation*}
where $W$ and $\eta$ are two large parameters.
Ideally, one would like to study the long time limit $\eta\to\infty$ for a fixed $W$. In this case, however, we know 
that the dynamics cannot be diffusive for $d = 1$. Indeed, as explained in the introduction, it is expected that the 
motion cannot be diffusive for distances larger than $W^2$; this has in fact been proved \cite{Sch} for distances larger 
than $W^8$.  Thus we have to consider a scaling limit where $\eta$ and $W$
are related and they tend simultaneously to infinity. To that end we choose an exponent $\kappa>0$
and set $\eta \equiv \eta(W) \deq W^{d \kappa}$.

Our first main result establishes that $\varrho(t,x)$ behaves diffusively up to time scales $t = 
O(W^{d \kappa})$ if $\kappa<1/3$. 

\begin{theorem}[Quantum diffusion] \label{theorem: first main result}
Let $0 < \kappa < 1/3$ be fixed. Then for any $T_0 > 0$ and any continuous bounded function 
$\varphi \in C_b(\R^d)$ we have
\begin{equation} \label{main result}
\lim_{W \to \infty} \sum_{x\in \Lambda_N}
 \varrho\pb{W^{d \kappa} T, x} \, \varphi \pbb{\frac{x}{W^{1 + d \kappa/2}}} \;=\; \int_{\R^d} \dd 
X \; L(T, X) \, \varphi(X)\,,
\end{equation}
uniformly in $N \geq W^{1 + d/6}$ and $0 \leq T \leq T_0$. Here
\begin{equation*}
L(T,X) \;\deq\; \int_0^1 \dd \lambda \; \frac{4}{\pi} \frac{\lambda^2}{\sqrt{1 - \lambda^2}}
\, G(\lambda T, X)\,,
\end{equation*}
and $G$ is the heat kernel
\begin{equation} \label{definition of heat kernel}
G(T, X) \;\deq\; \pbb{\frac{d + 2}{2 \pi T}}^{d/2} \, \ee^{- \frac{d+2}{2 T} \, \abs{X}^2}\,,
\end{equation}
\end{theorem}

\begin{remark}
The factor $d + 2$ arises from a random walk in $d$ dimensions with
 steps in the unit ball.  If $B$ is a random variable 
uniformly distributed in the $d$-dimensional unit ball, the covariance
 matrix of $B$ is $(d+2)^{-1} \umat$.
\end{remark}

This result can be interpreted as follows.
The limiting dynamics at macroscopic time $T$ is not given by a single
 heat kernel, but by a weighted superposition of heat kernels at times $\lambda T$, for $0 \leq \lambda \leq 1$. The 
factor $\lambda$ expresses a delay arising from backtracking paths, in which the quantum particle ``wastes time'' by
 retracing its steps. If the particle is not backtracking, it is moving according to diffusive dynamics. The 
backtracking paths correspond to two-legged subdiagrams, and have the interpretation of a self-energy renormalization in 
the language of diagrammatic perturbation theory.  Thus, out of the total macroscopic time $T$ during which the particle 
moves, a fraction $\lambda$ of $T$
 is spent moving diffusively, and a fraction $(1 - \lambda)$ of $T$ backtracking. Theorem~\ref{theorem: first main 
result} gives an explicit expression for the probability density $f(\lambda) =
 \frac{4}{\pi}\frac{\lambda^2}{\sqrt{1 - \lambda^2}} \ind{0 \leq \lambda \leq 1}$ for the particle to move during a 
fraction $\lambda$ of $T$.

Our proof  precisely exhibits this phenomenon. As explained in Section~\ref{ideas of proof}, the proof is based on an 
expansion of the quantum time evolution in terms of \emph{nonbacktracking paths}.
 At time $t = W^{d \kappa} T$, this 
expansion yields a weighted superposition of paths of lengths
 $n = 1, \dots, [t]$ (higher values of $n$ are strongly 
suppressed). Here $n$ is the number of nonbacktracking steps, i.e.\ the
 number of steps that contribute to the effective motion of the particle. The difference $[t] - n$ is the number of 
steps that the particle
 spends backtracking.  Our expansion (or, 
more precisely, its leading order \emph{ladder terms}) shows that the weight 
of a path of $n$ nonbacktracking steps is 
given by
$\abs{\alpha_n(t)}^2$, where $\alpha_n(t)$ is the Chebyshev transform of the
 propagator $\ee^{- \ii t \xi}$ in $\xi$; 
see \eqref{formula for alpha}.
The probability density $f$ arises from this microscopic picture by
setting $n = [\lambda t]$. Then we have, as proved in Proposition~\ref{proposition: asymptotics of alpha_n} below,
$t \abs{\alpha_{[\lambda t]}(t)}^2 \to f(\lambda)$
weakly as $t \to \infty$.

Our second main result shows that the eigenvectors of $H$ have a typical localization length 
larger than $W^{1 + d \kappa / 2}$, for any $\kappa < 1/3$. For $x \in \Lambda_N$ and $\ell > 
0$ we define the characteristic function  $P_{x,\ell}$ projecting onto the complement of an 
$\ell$-neighbourhood
of $x$,
\begin{equation*}
P_{x, \ell}(y) \;\deq\; \ind{\abs{y - x} \geq \ell}\,.
\end{equation*}
Let $\epsilon > 0$ and define the random subset $\fra A_{\epsilon, \ell}^\omega \subset \fra A$ 
of eigenvectors through
\begin{equation*}
\fra A^\omega_{\epsilon, \ell} \;\deq\; \hbb{\alpha \in \fra A \,:\, \sum_x 
\abs{\psi_\alpha^\omega(x)} \, \norm{P_{x, \ell} \, \psi^\omega_\alpha} < \epsilon}\,.
\end{equation*}
The set $\fra A^\omega_{\epsilon, \ell}$ contains, in particular, all eigenvectors that are exponentially localized in 
balls of radius $O(\ell)$; see Corollary \ref{cor: deloc} below for a more general and precise statement.


\begin{theorem}[Delocalization] \label{theorem: second main result}
Let $\epsilon > 0$ and $0 < \kappa < 1/3$. Then
\begin{equation*}
\limsup_{W \to \infty} \, \E \, \frac{\absb{\fra A_{\epsilon, W^{1 + d \kappa /2}}}}{\abs{\fra A}} \;\leq\; 2 
\sqrt{\epsilon}\,,
\end{equation*}
uniformly in $N \geq W^{1 + d / 6}$.
\end{theorem}

Theorem \ref{theorem: second main result} implies that the fraction of eigenvectors subexponentially localized on scales 
$W^{1 + \kappa d / 2}$ converges to zero in probability.

\begin{corollary} \label{cor: deloc}
For fixed $\gamma > 0$ and $K > 0$ define the random subset of eigenvectors
\begin{equation} \label{definition of B_l}
\fra B^\omega_\ell \;\deq\; \hbb{\alpha \in \fra A \,:\, \exists \, u \in \Lambda_N \,:\, \sum_x 
\abs{\psi_\alpha^\omega(x)}^2 \exp \qbb{\frac{\abs{x - u}}{\ell}}^\gamma \leq K}\,.
\end{equation}
Then for $0 < \kappa < 1/3$ we have
\begin{equation*}
\lim_{W \to \infty} \E \, \frac{\abs{\fra B_{W^{1 + \kappa d / 2}}}}{\abs{\fra A}} \;=\; 0\,,
\end{equation*}
uniformly in $N \geq W^{1 + d/6}$.
\end{corollary}

\section{Main ideas of the proof} \label{ideas of proof}

We need to compute the expectation of the squared matrix elements of
the unitary time evolution $\ee^{-\ii tH/2}$. A natural starting point is the power series expansion $\ee^{-\ii tH/2}
=\sum_{n\geq 0} (-\ii tH/2)^n/n!$. Unfortunately, the resulting series is \emph{unstable} for $t \to \infty$, as is 
manifested
by the large cancellations in the sum
\begin{equation}
\E \abs{\scalar{\delta_x}{\ee^{-\ii t H/2} \delta_y}}^2 = \sum_{n,n'} \frac{\ii^{n - n'} 
t^{n+n'}}{2^{n+ n'} n!  n'!} \, \E H^n_{xy}H^{n'}_{yx}.  \label{EE}
\end{equation}
This can be seen as follows. The expectation
\begin{equation}
    \E H^n_{xy}H^{n'}_{yx} = \E \sum_{x_1, \ldots x_{n-1}}\sum_{y_1, \ldots y_{n'-1}}
	H_{xx_1}H_{x_1x_2}\ldots H_{x_{n-1}y} H_{yy_{n' - 1}}\ldots H_{y_1 x}
\label{EhEh}
\end{equation}
is traditionally represented graphically by drawing the labels $x, x_1, x_2, \dots, y_1, x$ as vertices of a path, and 
by identifying vertices whose labels are identical. Since the matrix elements are centred (i.e.\ $\E H_{xy}$ = 0 for all 
$x,y$), each edge must be traveled at least twice in any path that yields a nonzero contribution to \eqref{EhEh}. It is 
well known that the leading order contribution to \eqref{EhEh} is given by the so-called \emph{fully backtracking 
paths}. A fully backtracking path is a path generated by successively applying the transformation $a \mapsto aba$ to the 
trivial path $x$. A typical fully backtracking path may be thought of as a tree with double edges. It is not hard to see 
that, after summing over $y$, each fully backtracking path yields a contribution of order 1 to \eqref{EhEh}. Also, the 
number of fully backtracking paths is of order $4^{n + n'}$, so that the expectation \eqref{EhEh} is of order 
$4^{n+n'}$. In particular, this implies that the main contribution to \eqref{EE} comes from terms satisfying $n+n' \sim 
t$. Moreover, the series \eqref{EE} is unstable in the sense that the sum of the absolute values of its summands behaves 
like $\ee^{4t}$ as $t \to \infty$.

The large terms in \eqref{EE} systematically
cancel each other out similarly to the {\it two-legged subdiagram renormalization}
in perturbative field theory.  In perturbative renormalization, 
these cancellations are exploited by introducing appropriately adjusted
 fictitious counter-terms. In the current problem, however, we make use of the \emph{Chebyshev 
transformation}, which removes the contribution of all backtracking paths in one step. The key 
observation is that, if $U_n$ denotes the $n$-th
Chebyshev polynomial of the second kind, then  $U_n(H/2)$ can be expressed
in terms of \emph{nonbacktracking paths}. A nonbacktracking path is a path which contains no subpath of the form $aba$.  
Thus the
strongest instabilities in \eqref{EE} can be removed if $\ee^{-\ii tH/2}$ is
expanded into a series of Chebyshev polynomials. 
This idea appeared first in \cite{BY} and has recently been exploited in \cite{FSo,So1} to 
prove, among other things, the edge-universality
for band matrices. In \cite{So1} it is also stated that
the same method can be used to prove delocalization of the edge
eigenvectors if $W \geq N^{5/6}$, i.e.\ to get the bound $\ell \geq W^{6/5}$ on the localization length $\ell$.  Our 
estimate gives a slightly weaker bound, $\ell \geq W^{7/6}$,
for this special case, but it applies to bulk eigenvectors as well as higher dimensions.

After the Chebyshev transform, we need to compute expectations
\begin{equation}
    \E \sideset{}{'}\sum_{x_1, \ldots x_{n-1}}
\sideset{}{'}\sum_{y_1, \ldots y_{n'-1}}
	H_{xx_1}H_{x_1x_2}\ldots H_{x_{n-1}y} H_{yy_{n' - 1}}\ldots H_{y_1x} ,
\label{hprime}
\end{equation}
where the summations are restricted to nonbacktracking paths. As above, since $\E H_{ab}=0$ every matrix element must 
appear at least twice in the non-trivial terms of \eqref{hprime}.
Taking the expectation effectively introduces a pairing, or more generally
a lumping, of the factors, which can be conveniently
represented by  Feynman diagrams. The main contribution comes from the so-called
{\it ladder diagrams}, corresponding to $n=n'$ and $x_i=y_i$. The contribution of these 
diagrams can be explicitly computed, and showed to behave diffusively. More precisely: Since we 
express nonbacktracking powers of $H$ as Chebyshev polynomials in $H/2$,
 the contribution of each 
graph to the propagator $\ee^{-\ii t H/2}$ carries a weight equal to the
 Chebyshev transform $\alpha_n(t)$ of 
$\ee^{-\ii t \xi}$ in $\xi$. We shall show that $\alpha_n(t)$ is given 
essentially by a Bessel function of the first
kind. In order to identify the limiting behaviour of the ladder diagrams, we therefore need 
to analyse a probability distribution on $\N$ of the form $\hb{\abs{\alpha_n(t)}^2}_{n \in \N}$ for 
large $t$ (Section~\ref{section: ladder}).

The main work consists of proving that the non-ladder diagrams are negligible.
Similarly to the basic idea of \cite{ESY1, ESY2, ESY3}, the non-ladder
diagrams are classified according to their combinatorial complexity.
The large number of complex diagrams is offset by their
small value, expressed in terms of powers of $W$. Conversely,
diagrams containing large pieces of ladder subdiagrams have a relatively
large contribution but their number is small.

More precisely, focusing only on the pairing diagrams in the Hermitian case, it is easy to see that ladder subdiagrams 
are marginal for power counting. We define
the {\it skeleton} of a graph by collapsing parallel ladder rungs (called {\it bridges})
into a single rung. We show that the value of a skeleton diagram is
given by a negative power of $M \sim C W^d$ that is proportional to the size of the skeleton diagram.  This is how the 
dimension $d$ enters our estimate. We then
sum up all possible ladder subdiagrams corresponding to a given skeleton.
Although the ladder subdiagrams do not yield additional $W$-powers,
they represent classical random walks for which dispersive bounds
are available, rendering them summable.
The restriction $t\ll W^{d/3}$ comes from summing up the skeleton
diagrams. In Section \ref{sec:further} we present a critical skeleton
that shows that this restriction is necessary without further resummation
or a more refined classification of complex graphs.

\section{The path expansion} \label{section: path expansion}
We start by writing the expansion of $\ee^{-\ii t H/2}$ in terms of nonbacktracking paths by 
using the Chebyshev transform.

\subsection{The Chebyshev transform of $\ee^{-\ii t \xi}$}
The Chebyshev transform $\alpha_k(t)$ of $\ee^{-\ii t \xi}$ is defined by
\begin{equation*}
\ee^{-\ii t \xi} \;=\; \sum_{k = 0}^\infty \alpha_k(t) \, U_k(\xi)\,.
\end{equation*}
Here $U_k$ denotes the Chebyshev polynomial of the second kind, defined through
\begin{equation} \label{definition of Chebyshev polynomial}
U_k(\cos \theta) \;=\; \frac{\sin (k+1) \theta}{\sin \theta}
\end{equation}
for $k = 0,1,2,\dots$.
The Chebyshev polynomials satisfy the orthogonality relation
\begin{equation*}
\frac{2}{\pi} \int_{-1}^1 \dd \xi \, \sqrt{1 - \xi^2} \,  U_k(\xi) \, U_l(\xi) \;=\; 
\delta_{kl}\,.
\end{equation*}
Therefore the coefficients $\alpha_k(t)$ are given by
\begin{equation} \label{definition of the coefficients alpha}
\alpha_k(t) \;=\; \frac{2}{\pi} \int_{-1}^1 \dd \xi \; \sqrt{1 - \xi^2} \, \ee^{- \ii t \xi} \, 
U_k(\xi)\,.
\end{equation}

The coefficient $\alpha_k(t)$ can be evaluated explicitly using the standard identities (see 
\cite{GradshteynRyzhik})
\begin{gather*}
U_k(\xi) \;=\; \frac{\xi T_{k+1}(\xi) - T_{k+2}(\xi)}{1 - \xi^2}\,,
\qquad
T_{k+2}(\xi) - 2\xi T_{k+1}(\xi) + T_k(\xi) \;=\; 0\,,
\\
\frac{2}{\pi} \int_{-1}^1 \dd \xi \; \frac{T_{2l}(\xi) \cos(t \xi)}{\sqrt{1 - \xi^2}} \;=\; 2 (-1)^l J_{2l}(t)\,,
\qquad
\frac{2}{\pi} \int_{-1}^1 \dd \xi \; \frac{T_{2l+1}(\xi) \sin(t \xi)}{\sqrt{1 - \xi^2}} \;=\; 2 (-1)^l J_{2l+1}(t)\,,
\\
J_{k}(t) + J_{k + 2}(t) \;=\; \frac{2(k+1)}{t} J_{k + 1}(t)\,.
\end{gather*}
Here $T_k$ denotes the Chebyshev polynomial of the first kind and $J_k$ the Bessel function of the first kind; they are 
defined through
\begin{equation*}
T_k(\cos \theta) \;\deq\; \cos (k \theta)\,,
\qquad
J_k(t) \;\deq\; \frac{1}{\pi} \int_0^\pi \dd \theta \; \cos (t \sin \theta - t \theta)\,.
\end{equation*}

If $k = 2l$ is even we may therefore compute
\begin{multline*}
\alpha_{2l}(t) \;=\; \frac{2}{\pi} \int_{-1}^1 \dd \xi \; \sqrt{1 - \xi^2} \, \cos (t\xi) \, U_k(\xi)
\;=\; \frac{2}{\pi} \int_{-1}^1 \dd \xi \; \sqrt{1 - \xi^2} \, \cos (t\xi) \, \frac{\xi T_{k+1}(\xi) - T_{k+2}(\xi)}{1 - 
\xi^2}
\\
\;=\; \frac{2}{\pi} \int_{-1}^1 \dd \xi \; \cos (t\xi) \, \frac{T_k(\xi) - T_{k+2}(\xi)}{2\sqrt{1 - \xi^2}}
\;=\; (-1)^l \qb{J_{2l}(t) + J_{2l + 2}(t)}
\;=\; 2 (-1)^l \, \frac{2l+1}{t} J_{2l + 1}(t)\,.
\end{multline*}
If $k = 2l + 1$ is odd a similar calculation yields
\begin{equation*}
\alpha_{2l + 1}(t) \;=\; -2 \ii (-1)^l \, \frac{2l+2}{t} J_{2l + 2}(t)\,.
\end{equation*}
Thus we have the following result.
\begin{lemma} \label{lemma: Chebyshev transform of exponential}
We have that
\begin{equation*}
\ee^{-\ii t \xi} \;=\; \sum_k \alpha_k(t) \, U_k(\xi)\,,
\end{equation*}
where
\begin{equation} \label{formula for alpha}
\alpha_k(t) \;=\; 2 (-\ii)^k \frac{k+1}{t} \, J_{k + 1}(t)\,.
\end{equation}
\end{lemma}

Also, for all $t \in \R$ we have the identity
\begin{equation} \label{orthonormality of Chebyshev coefficients}
\sum_{k \geq 0} \abs{\alpha_k(t)}^2 \;=\; 1\,,
\end{equation}
as follows from the orthonormality of the Chebyshev polynomials.

\subsection{Expansion in terms of nonbacktracking paths}
For $n = 0, 1, 2, \dots$ let $H^{(n)}$ denote the $n$-th nonbacktracking power of $H$.  It is 
defined by
\begin{equation*}
H^{(n)}_{x_0, x_n} \;\deq\; \sideset{}{'}\sum_{x_1, \dots, x_{n-1}} H_{x_0 x_1} \cdots 
H_{x_{n-1} x_n}\,,
\end{equation*}
where $\sum'$ means sum under the restriction $x_i \neq x_{i+2}$ for $i = 0, \dots, n - 2$. We 
call this restriction the \emph{nonbacktracking condition}.

The following key observation is due to Bai and Yin \cite{BY}.
\begin{lemma} \label{lemma recursion for nonbacktracking random walks}
The nonbacktracking powers of $H$ satisfy
\begin{equation*}
H^{(0)} \;=\; \umat \,, \qquad H^{(1)} \;=\; H \,, \qquad H^{(2)} \;=\; H^2 - \frac{M}{M - 1} 
\umat\,,
\end{equation*}
as well as the recursion relation
\begin{equation} \label{recursion for nonbacktracking paths}
H^{(n)} \;=\; H H^{(n-1)} - H^{(n-2)} \qquad (n \geq 3)\,.
\end{equation}
\end{lemma}
\begin{proof}
For the convenience of the reader we give the simple proof. The cases $n = 0,1,2$ are easily 
checked. Moreover,
\begin{align*}
(H H^{(n - 1)})_{x_0 x_n} &\;=\; \sum_{x_1, \dots, x_{n - 1}} \prod_{i = 1}^{n - 2} \ind{x_i 
\neq x_{i+2}} \, H_{x_0 x_1} \cdots H_{x_{n - 1} x_n}
\\
&\;=\; \sum_{x_1, \dots, x_{n - 1}} \prod_{i = 0}^{n - 2} \ind{x_i \neq x_{i+2}} \, H_{x_0 x_1} 
\cdots H_{x_{n - 1} x_n}
\\
&\qquad +
\sum_{x_1, \dots, x_{n - 1}} \ind{x_0 = x_2} \prod_{i = 1}^{n - 2} \ind{x_i \neq x_{i+2}} \, 
H_{x_0 x_1} \cdots H_{x_{n - 1} x_n}
\\
&\;=\; (H^{(n)})_{x_0, x_n} + \sum_{x_3, \dots, x_{n - 1}} \ind{x_0 \neq x_4} \prod_{i = 3}^{n 
- 2} \ind{x_i \neq x_{i + 2}} H_{x_0 x_3} H_{x_3 x_4} \cdots H_{x_{n - 1} x_n}
\\
&\qquad \qquad \qquad \qquad \qquad  \times
 \sum_{x_1} \ind{x_1 \neq x_3} \abs{H_{x_0 x_1}}^2
\\
&\;=\; (H^{(n)})_{x_0, x_n} + (H^{(n - 2)})_{x_0 x_n}\,.
\end{align*}
Notice that in the last step we used \eqref{deterministic norm}.
\end{proof}

Feldheim and Sodin have observed \cite{FSo,So1} that \eqref{recursion for nonbacktracking 
paths} is reminiscent of the recursion relation for the Chebyshev polynomials of the second 
kind. Let us abbreviate
$\widetilde{U}_n(\xi) \deq U_n(\xi / 2)$.
Then we have (see e.g.\ \cite{GradshteynRyzhik})
\begin{equation*}
\widetilde U_0(\xi) \;=\; 1\,, \qquad \widetilde U_1(\xi) \;=\; \xi \,, \qquad \widetilde 
U_2(\xi) \;=\; \xi^2 - 1\,,
\end{equation*}
and for $n \geq 2$
\begin{equation*}
\widetilde{U}_n(\xi) \;=\; \xi \widetilde{U}_{n-1}(\xi) - \widetilde{U}_{n-2}(\xi)\,.
\end{equation*}
Comparing this to Lemma \ref{lemma recursion for nonbacktracking random walks}, we get, 
following \cite{FSo, So1},
\begin{equation*}
H^{(n)} \;=\; \widetilde U_n(H) - \frac{1}{M - 1} \widetilde U_{n - 2}(H)\,.
\end{equation*}
Solving for $\widetilde U_n(H)$ yields
\begin{equation*}
\widetilde U_n(H) \;=\; \sum_{k \geq 0} \frac{1}{(M - 1)^k} H^{(n - 2k)}\,,
\end{equation*}
with the convention that $H^{(n)} = 0$ for $n < 0$.
Therefore Lemma \ref{lemma: Chebyshev transform of exponential} yields
\begin{equation*}
\ee^{- \ii t H/2} \;=\; \sum_{n \geq 0} \alpha_n(t) \, \widetilde U_n(H) \;=\;  \sum_{m \geq 0} 
H^{(m)} \sum_{k \geq 0} \frac{\alpha_{m + 2k}(t)}{(M - 1)^k}\,.
\end{equation*}

We have proved the following result.
\begin{lemma}\label{lemma: nonbacktracking path expansion}
We have that
\begin{equation*}
\ee^{- \ii t H/2} \;=\; \sum_{m \geq 0} a_m(t) H^{(m)}\,,
\end{equation*}
where
\begin{equation*}
a_m(t) \;\deq\; \sum_{k \geq 0} \frac{\alpha_{m + 2k}(t)}{(M - 1)^k} \,.
\end{equation*}
\end{lemma}

\section{Graphical representation} \label{section: graphical representation}
For ease of presentation, we assume throughout the proof of Theorem \ref{theorem: first main 
result} (Sections \ref{section: graphical representation} -- \ref{section: ladder}) that we are 
in the Hermitian case \eqref{hnormalization_1}. How to extend our arguments to cover the 
symmetric case \eqref{hnormalization_2} is described in Section \ref{section: symmetric 
matrices}.

Using Lemma \ref{lemma: nonbacktracking path expansion} we get
\begin{equation*}
\varrho(t,x) \;=\; \sum_{n, n' \geq 0} a_{n}(t) \ol{a_{n'}(t)} \, \E \, H^{(n)}_{0 x} H^{(n')}_{x 0}\,.
\end{equation*}
Expanding in nonbacktracking paths yields a graphical expansion. Let us write
$H^{(n)}_{0 x} H^{(n')}_{x 0}$
as a sum over paths $x_0, x_1, \dots, x_{n + n' - 1}, x_0$, where $x_0 = 0$ and $x_n = x$. Such a path is graphically 
represented as a loop of $n + n'$ vertices belonging to the set $\cal V_{n,n'} \deq \{0, \dots, n + n' - 1\}$; see 
Figure \ref{figure: basic graph}. Vertices $i \in \cal V_{n,n'}$ satisfying the nonbacktracking condition (i.e.\ $x_{i 
-1} \neq x_{i+1}$) are drawn using black dots; other vertices are drawn using white dots.
\begin{figure}[ht!]
\begin{center}
\includegraphics{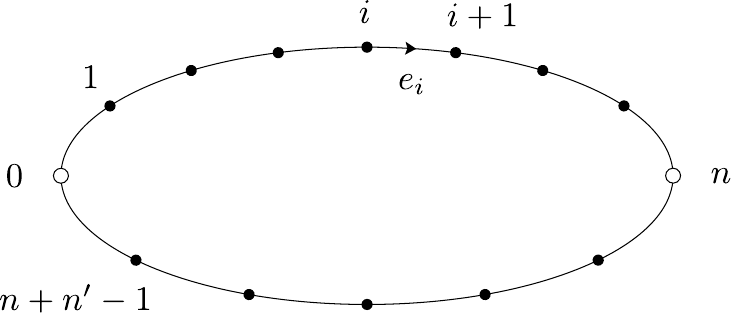}
\end{center}
\caption{The graphical representation of paths of vertices. \label{figure: basic graph}}
\end{figure}
There are $n + n'$ oriented edges $e_0, \dots, e_{n + n' - 1}$ defined by $e_i \deq (i, i + 1)$ (here, and in the 
following, $\cal V_{n,n'}$ is taken to be periodic). We denote by $\cal E_{n,n'} \deq \{e_0, \dots, e_{n+n'-1}\}$ the 
set of edges. In Figure \ref{figure: basic graph} the edges are oriented clockwise.  Each vertex has an outgoing and an 
incoming edge, and each edge $e$ has an initial vertex $a(e)$ and final vertex $b(e)$.  Moreover, we order the edges 
using their initial vertices.

Each vertex $i \in \cal V_{n,n'}$ carries a label $x_i \in \Lambda_N$. The labels $\b x = (x_0, \dots, x_{n + n' - 1})$ 
are summed over under the restriction $Q_x(\b x) = 1$, where
\begin{equation*}
Q_x(\b x) \;\deq\; \delta_{0 x_0} \delta_{x_n x} \prod_{i = 0}^{n + n' - 1} \ind{1 \leq \abs{x_i - x_{i + 1}} \leq W} 
\prod_{i = 0}^{n -2} \ind{x_i \neq x_{i+2}} \prod_{i = n}^{n + n' -2} \ind{x_i \neq x_{i+2}}\,.
\end{equation*}
The two last products implement the nonbacktracking condition. We define the unordered pair of labels corresponding to 
the edge $e$ through
\begin{equation*}
\varrho_{\b x}(e) \;\deq\; \{x_{a(e)}, x_{b(e)}\}\,.
\end{equation*}

Next, to each configuration of labels $\b x = (x_0, \dots, x_{n + n' - 1})$ we assign a \emph{lumping} $\Gamma = 
\Gamma(\b x)$ of the set of edges $\cal E_{n,n'}$. Here a lumping means a partition of $\cal E_{n,n'}$ or, equivalently, 
an equivalence relation on $\cal E_{n,n'}$. We use the notation $\Gamma = \{\gamma\}_{\gamma \in \Gamma}$, where $\gamma 
\in \Gamma$ is \emph{lump} of $\Gamma$, i.e.\ an equivalence class. The lumping $\Gamma = \Gamma(\b x)$ associated with 
the labels $\b x$ is defined according to the rule that $e$ and $e'$ are in the same lump $\gamma \in \Gamma$ if and 
only if $\varrho_{\b x}(e) = \varrho_{\b x}(e')$.
Let $\widetilde{\scr G}_{n, n'}$ denote the set of lumpings of $\cal E_{n,n'}$ obtained in this manner. Thus we may 
write
\begin{equation*}
\E \, H^{(n)}_{0 x} H^{(n')}_{x 0} \;=\; \sum_{\Gamma \in \widetilde{\scr G}_{n,n'}} V_x(\Gamma) \,.
\end{equation*}
Here
\begin{equation*}
V_x(\Gamma) \;=\; \sideset{}{^*}\sum_{\b x} Q_x(\b x) \, \E \, H_{x_0 x_1} \cdots H_{x_{n + n' - 1} x_0}\,,
\end{equation*}
where the summation is restricted to label configurations yielding the lumping $\Gamma$.

Next, observe that the expectation of a monomial $\prod_{y,z} (H_{yz})^{\nu_{yz}}$ is nonzero if and only if $\nu_{yz} = 
\nu_{zy}$ for all $y,z$ (here we only use that the law of the matrix entries is invariant under rotations of the complex 
plane). In particular, $V_x(\Gamma)$ vanishes if one lump $\gamma \in \Gamma$ is of odd size.  Defining the subset $\scr 
G_{n,n'} \subset \widetilde{\scr G}_{n,n'}$ of lumpings whose lumps are of even size, we find that
\begin{equation*}
\E \, H^{(n)}_{0 x} H^{(n')}_{x 0} \;=\; \sum_{\Gamma \in \scr G_{n,n'}} V_x(\Gamma) \,.
\end{equation*}

We summarize the key properties of $\scr G_{n,n'}$.

\begin{lemma} \label{lemma: characterization of the set of partitions}
Let $\Gamma \in \scr G_{n,n'}$. Then each lump $\gamma \in \Gamma$ is of even size. Moreover, any two edges $e,e' \in 
\gamma$ in the same lump $\gamma$ are separated by either at least two edges or a vertex in $\{0, n\}$ (nonbacktracking 
property).
\end{lemma}

Next, we give an explicit expression for $V_x(\Gamma)$. We start by assigning to each lump $\gamma \in \Gamma$ an 
unordered pair of labels $\varrho_\gamma$. Then we pick a partition $\pi_\gamma$ of $\gamma$ into two subsets of equal 
size. Abbreviate these families as $\b \varrho_\Gamma = \{\varrho_\gamma\}_{\gamma \in \Gamma}$ and $\b \pi_\Gamma = 
\{\pi_\gamma\}_{\gamma \in \Gamma}$.
Thus we get
\begin{equation} \label{expression for V(Gamma)}
V_x(\Gamma) \;=\; \sum_{\b x} Q_x(\b x) \sum_{\b \varrho_\Gamma} \sum_{\b \pi_\Gamma} \pBB{\prod_{\gamma \in \Gamma} 
\Delta_{\b x}(\varrho_\gamma, \pi_\gamma)} \pBB{\prod_{\gamma \neq \gamma'} \ind{\varrho_\gamma \neq \varrho_{\gamma'}}} 
\, \E H_{x_0 x_1} \cdots H_{x_{n + n' - 1} x_0}
\,.
\end{equation}
Here, for each $\gamma \in \Gamma$, $\varrho_\gamma$ ranges over all unordered pairs of labels and $\pi_\gamma$ ranges 
over all partitions of $\gamma$ into two subsets of equal size; $\Delta_{\b x}(\varrho_\gamma, \pi_\gamma)$ is the 
indicator function of the following event: For all $e \in \gamma$ we have that $\varrho_{\b x}(e) = \varrho_\gamma$, and
\begin{align*}
\text{$e,e' \in \gamma$ belong to the same subset of $\pi_\gamma$} &\qquad \Longrightarrow \qquad x_{a(e)} \;=\; 
x_{a(e')}\,, \quad x_{b(e)} \;=\; x_{b(e')}
\\
\text{$e,e' \in \gamma$ belong to different subsets of $\pi_\gamma$} &\qquad \Longrightarrow \qquad x_{a(e)} \;=\; 
x_{b(e')}\,, \quad x_{b(e)} \;=\; x_{a(e')}\,.
\end{align*}
This definition of $\Delta_{\b x}(\varrho_\gamma, \pi_\gamma)$ has the following interpretation.  All edges in $\gamma$ 
(corresponding to matrix elements) have the same unordered pair of labels (and hence represent copies of the same random 
variable $H_{yz}$ or its complex conjugate).  Moreover, each random variable $H_{yz}$ must appear as many times as its 
complex conjugate; random variables indexed by two edges $e,e' \in \gamma$ are identical if $e,e'$ belong to the same 
subset of $\pi_\gamma$, and each other's complex conjugates if $e,e'$ belong to different subsets of $\pi_\gamma$.

Note that the expectation in \eqref{expression for V(Gamma)} is equal to
\begin{equation} \label{expectation of monomial}
\frac{1}{(M - 1)^{\bar n}}\,,
\end{equation}
where
$\bar n \deq \frac{n + n'}{2}$.
In particular, $V_x(\Gamma) \geq 0$.

\begin{figure}[ht!]
\begin{center}
\includegraphics{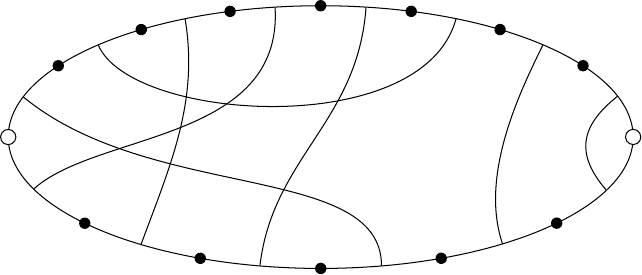}
\end{center}
\caption{A pairing of edges.\label{figure: pairings}}
\end{figure}
An important subset of lumpings of $\cal E_{n,n'}$ is the set of pairings, $\scr P_{n,n'} \subset \scr G_{n,n'}$, which 
contains all lumpings $\Gamma$ satisfying $\abs{\gamma} = 2$ for all $\gamma \in \Gamma$.  We call two-element lumps 
$\sigma \in \cal P_{n,n'}$ \emph{bridges}.  Given a pairing $\Gamma \in \scr P_{n,n'}$, we say that $e$ and $e'$ are 
\emph{bridged} (in $\Gamma$) if there is a $\sigma \in \Gamma$ such that $\sigma = \{e, e'\}$.  Bridges are represented 
graphically by drawing a line, for each $\{e,e'\} \in \Gamma$, from the edge $e$ to $e'$; see Figure \ref{figure: 
pairings}.  Thus a pairing $\Gamma \in \scr P_{n,n'}$ is the edge set of a graph whose vertex set is $\cal E_{n,n'}$. If 
$\Gamma$ is a pairing, each bridge $\sigma \in \Gamma$ has a unique partition $\pi_\sigma$ of its edges, so that the 
expression \eqref{expression for V(Gamma)} for $V_x(\Gamma)$ may be rewritten in the simpler form
\begin{multline} \label{expression for V for pairings}
V_x(\Gamma) \;=\; \sum_{\b x} Q_x(\b x) \pBB{\prod_{\{e, e'\} \in \Gamma} \ind{x_{a(e)}  = x_{b(e')}} \ind{x_{b(e)} =  
x_{a(e')}} }
\\
\times \pBB{\prod_{\sigma \neq \sigma'} \prod_{e \in \sigma} \prod_{e' \in \sigma'} \ind{\varrho_{\b x}(e) \neq 
\varrho_{\b x}(e')}} \frac{1}{(M - 1)^{\bar n}} \,.
\end{multline}

The main contribution to the expansion is given by the \emph{ladder pairing} $L_n \in \scr P_{n,n}$.  It is defined as
\begin{equation*}
L_n \;\deq\; \hb{\{e_0, e_{2n - 1}\}, \{e_1, e_{2n - 2}\}, \dots, \{e_{n - 1}, e_n\}}\,.
\end{equation*}
The ladder is represented graphically in Figure \ref{figure: ladder}.
\begin{figure}[ht!]
\begin{center}
\includegraphics{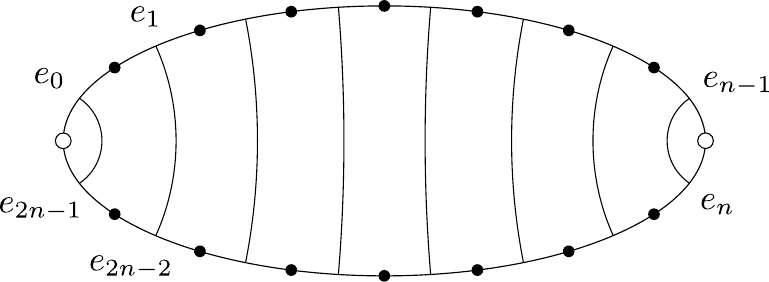}
\end{center}
\caption{The ladder pairing.\label{figure: ladder}}
\end{figure}

\section{The non-ladder lumpings} \label{section: error estimates}
In this section we estimate the contribution of the non-ladder lumpings and show that it vanishes in the limit $W \to 
\infty$. Let $\scr G_{n,n'}^* \subset \scr G_{n,n'}$ denote the set of non-ladder lumpings, i.e.\ $\scr G^*_{n,n'} \deq 
\scr G_{n,n'}$ if $n \neq n'$ and $\scr G^*_{n,n} \deq \scr G_{n,n} \setminus \{L_n\}$.  Similarly, let $\scr P_{n,n'}^* 
\deq \scr P_{n,n'} \cap \scr G^*_{n,n'}$ denote the set of non-ladder pairings.

We shall prove the following result.
\begin{proposition} \label{proposition: main estimate on error}
Let $0 < \kappa < 1/3$ and pick a $\beta$ satisfying $0 < \beta < 2/3 - 2\kappa$. Then there is a constant $C$ such that
\begin{equation*}
\sum_x \sum_{n,n \geq 0} \abs{a_n(\eta T) \, a_{n'}(\eta T)} \sum_{\Gamma \in \scr G^*_{n,n'}} V_x(\Gamma) \;\leq\; 
\frac{C}{W^{d \beta}}\,,
\end{equation*}
for $W$ larger than some $W_0(T, \kappa)$ and $N \geq W^{1 + d /6}$.
\end{proposition}
The rest of this section is devoted to the proof of Proposition \ref{proposition: main estimate on error}.

\subsection{Controlling the non-pairings}
Replacing the expectation in \eqref{expression for V(Gamma)} with \eqref{expectation of monomial} we get
\begin{equation*}
V_x(\Gamma) \;=\; \sum_{\b x} Q_x(\b x) \sum_{\b \varrho_\Gamma} \sum_{\b \pi_\Gamma} \pBB{\prod_{\gamma \in \Gamma} 
\Delta_{\b x}(\varrho_\gamma, \pi_\gamma)} \pBB{\prod_{\gamma \neq \gamma'} \ind{\varrho_\gamma \neq \varrho_{\gamma'}}} 
\, \frac{1}{(M - 1)^{\bar n}} \,.
\end{equation*}
We start by estimating the sum over all lumpings $\Gamma \in \scr G_{n,n'}^*$ in terms of a sum over all pairings 
$\Gamma \in \scr P_{n,n'}^*$. Let us define
\begin{equation} \label{definition of R}
R_x(\Gamma) \;\deq\; \sum_{\b x} Q_x(\b x) \sum_{\b \varrho_\Gamma} \sum_{\b \pi_\Gamma} \pBB{\prod_{\gamma \in \Gamma} 
\Delta_{\b x}(\varrho_\gamma, \pi_\gamma)} \, \frac{1}{(M - 1)^{\bar n}} \,.
\end{equation}

\begin{lemma} \label{lemma: sum over non-pairings}
For all $n,n' \in \N$ we have
\begin{equation*}
\sum_{\Gamma \in \scr G_{n,n'}^*} V_x(\Gamma) \;\leq\; \sum_{\Gamma \in \scr P_{n,n'}^*} R_x(\Gamma)\,.
\end{equation*}
\end{lemma}
\begin{proof}
Let $\varrho_\gamma$ and $\pi_\gamma$ be given for each $\gamma \in \Gamma$. For each $\gamma$, pick any pairing 
$\Sigma_\gamma$ of $\gamma$ that is compatible with $\pi_\gamma$ in the sense that, for each bridge $\sigma \in 
\Sigma_\gamma$, the two edges of $\sigma$ belong to different subsets of $\pi_\gamma$. If $n = n'$, we additionally 
require that not all $\Sigma_\gamma$'s are subsets of the Ladder $L_n$ (such a choice is always possible). Next, set 
$\varrho_\sigma = \varrho_\gamma$ for all $\sigma \in \Sigma_\gamma$. Note that each bridge $\sigma$ carries a unique 
partition $\pi_\sigma$. It is then easy to see that for any pairing $\Sigma_\gamma$ as above, we have
\begin{equation*}
\Delta_{\b x}(\varrho_\gamma, \pi_\gamma) \;\leq\; \prod_{\sigma \in \Sigma_\gamma} \Delta_{\b x} (\varrho_\sigma, 
\pi_\sigma)\,.
\end{equation*}
Thus, by partitioning each $\gamma \in \Gamma$ into bridges, we see that each term in $
\sum_{\Gamma \in \scr G_{n,n'}^*} V_x(\Gamma)$ is bounded by a corresponding term in
$\sum_{\Gamma \in \scr P_{n,n'}^*} R_x(\Gamma)$. In fact, there is an overcounting arising from the different ways of 
partitioning $\gamma$ into bridges.
\end{proof}

Because of Lemma \ref{lemma: sum over non-pairings} we may restrict ourselves to pairings. We 
estimate
$\sum_{\Gamma \in \scr P_{n,n'}^*} R_x(\Gamma)$.
If $\Gamma$ is a pairing we may write, just like \eqref{expression for V for pairings}, the expression \eqref{definition 
of R} in the simpler form
\begin{equation} \label{expression for R for pairings}
R_x(\Gamma) \;=\; \sum_{\b x} Q_x(\b x) \pBB{\prod_{\{e, e'\} \in \Gamma} \ind{x_{a(e)}  = 
x_{b(e')}} \ind{x_{b(e)} =  x_{a(e')}} } \, \frac{1}{(M - 1)^{\bar n}} \,.
\end{equation}

\subsection{Collapsing of parallel bridges}
Let us introduce the set $\ol{\scr P}^*_{n,n'}$, defined as the set of all non-ladder pairings 
of $\cal E_{n,n'}$. Clearly, $\scr P^*_{n,n'}$ is a proper subset of $\ol{\scr P}^*_{n,n'}$ 
(due to the nonbacktracking condition of Lemma \ref{lemma: characterization of the set of 
partitions} which is imposed on pairings in $\scr P^*_{n,n'}$).

Let $n,n' \geq 0$ and $\Gamma \in \ol{\scr P}^*_{n,n'}$. For any $i,j$, we say that the two bridges $\{e_i, e_j\}$ and 
$\{e_{i + 1}, e_{j - 1}\}$ of $\Gamma$ are \emph{parallel} if $i+1, j \notin \{0, n\}$; see Figure \ref{figure: parallel 
bridges}. Two parallel bridges may be collapsed to obtain a new pairing $\Gamma'$ of a smaller set of edges, in which 
the parallel bridges are replaced by a single bridge.  More precisely: We obtain $\Gamma' \in \ol{\scr P}^*_{m,m'}$ from 
$\Gamma \in \ol{\scr P}^*_{n,n'}$ by removing the vertices $i+1$ and $j$, by creating the edges $(i, i + 2)$ and $(j - 
1, j + 1)$, and by bridging them.  Finally, we rename the vertices using the increasing integers $0, 1,2,\dots, n+n' - 
3$; by definition, the new name of the vertex $n$ is $m$, and $m'$ is defined through $m+m'+2 = n+n'$.
\begin{figure}[ht!]
\begin{center}
\includegraphics{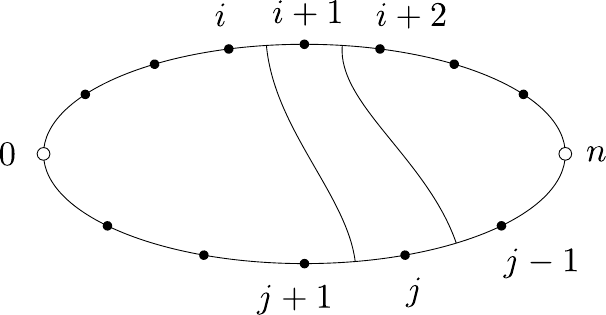}
\end{center}
\caption{Two parallel bridges.\label{figure: parallel bridges}}
\end{figure}
The converse operation of collapsing bridges, \emph{expanding bridges}, is self-explanatory.

In the next lemma we iterate the above procedure $\Gamma \mapsto \Gamma'$ until all parallel 
bridges have been collapsed. 

\begin{lemma} \label{lemma: collapsing bridges}
Let $\Gamma \in \scr P^*_{n,n'}$. Then there exist $m \leq n$, $m' \leq n'$, and a pairing
$S(\Gamma) \in \ol{\scr G}^*_{m,m'}$ containing no parallel bridges, such that $\Gamma$ may be 
obtained from $S(\Gamma)$ by successively expanding bridges. This defines $S(\Gamma)$ uniquely.  
\end{lemma}
\begin{proof}
Successively collapse all parallel bridges in $\Gamma$; see Figure \ref{figure: collapsing 
bridges}. The result is clearly independent of the order in which this is done.
\begin{figure}[ht!]
\begin{center}
\includegraphics{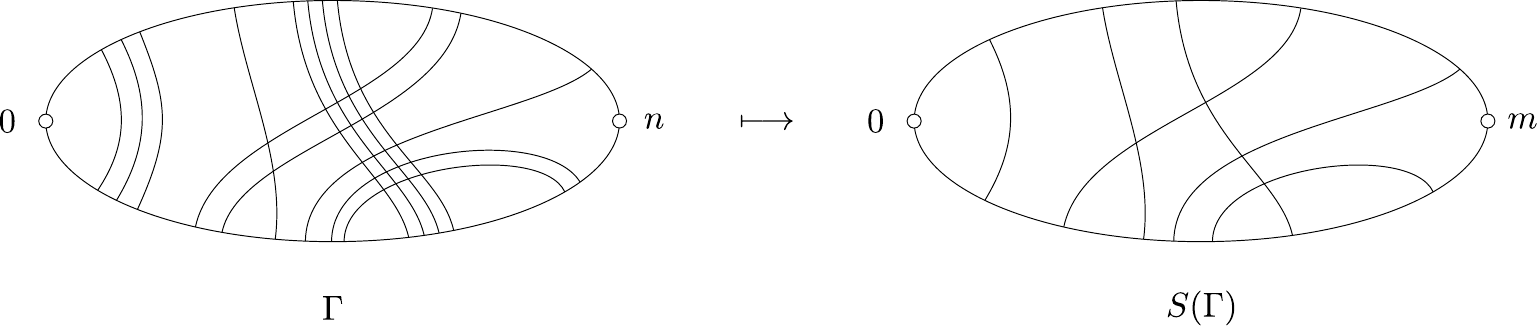}
\end{center}
\caption{Collapsing parallel bridges to obtain the skeleton pairing. \label{figure: collapsing 
bridges}}
\end{figure}
\end{proof}

We call the pairing $\Sigma = S(\Gamma)$ the \emph{skeleton} of $\Gamma$. The set of skeleton 
pairings of the edges $\cal E_{m,m'}$ is denoted by
\begin{equation*}
\scr S^*_{m,m'} \;\deq\; \bigcup_{n,n' \geq 0} \hb{S(\Gamma) \,:\, \Gamma \in \scr P^*_{n,n'}} 
\cap \ol{\scr P}^*_{m,m'}\,.
\end{equation*}
Note that $\scr S^*_{m,m'}$ is in general not a subset of $\scr P^*_{m,m'}$. The following 
lemma summarizes the key properties of $\scr S^*_{m,m'}$.

\begin{lemma} \label{lemma: properties of skeleton pairings}
\begin{enumerate}
\item
Each $\Sigma \in \scr S^*_{m,m'}$ contains no parallel bridges.
\item
Let $\Sigma \in \scr S^*_{m,m'}$ and $\sigma = \{e, e'\} \in \Sigma$. Then $e,e'$ are adjacent 
only if $e\cap e' \in \{0,m\}$.
\item
If $\bar m \deq \frac{m+m'}{2} = 1$ then $\scr S^*_{m,m'} = \emptyset$.
\end{enumerate}
\end{lemma}

\begin{proof}
Statement (i) follows immediately from the definition of $S(\Gamma)$. Statement (ii) is a 
consequence of the nonbacktracking property of pairings in $\scr P^*_{n,n'}$, i.e.\ Lemma 
\ref{lemma: characterization of the set of partitions}. To see this, let $\Sigma \in \scr 
S^*_{m,m'}$ be of the form $\Sigma = S(\Gamma)$ for some $\Gamma \in \scr P^*_{n,n'}$. If 
$\Sigma = S(\Gamma)$ contains a bridge $\{e,e'\}$ consisting of two consecutive edges $e,e'$, 
then $\Gamma$ must also contain a bridge $\{f,f'\}$ consisting of two consecutive edges $f,f'$.  
If $e \cap e' \notin \{0,m\}$, then $f \cap f' \notin \{0, n\}$, in contradiction to Lemma 
\ref{lemma: characterization of the set of partitions}. Statement (iii) is an immediate 
consequence of (ii) and the requirement that $L_1 \notin \scr S^*_{1,1}$.
\end{proof}

\subsection{Contribution of parallel bridges}
For given $n$ and $n'$, we estimate $\sum_{\Gamma \in \scr P^*_{n,n'}} R_x(\Gamma)$ by
summing over skeleton pairings $\Sigma$, followed by summing over all possible ways of 
expanding the bridges of $\Sigma$.

We observe that a pairing $\Gamma \in \scr P^*_{n,n'}$ is uniquely determined by its skeleton 
$\Sigma = S(\Gamma) \in \scr S^*_{m,m'}$ for some positive integers $m,m'$ as well as a family 
$\b \ell_\Sigma = \{\ell_\sigma\}_{\sigma \in \Sigma}$ satisfying $\abs{\b \ell_\Sigma} = \bar 
n$, where $\ell_\sigma$ encodes the number of parallel bridges that were collapsed to form the 
bridge $\sigma$. Here $\ell_\sigma \geq 1$ is a positive integer and $\abs{\b \ell_\Sigma} \deq 
\sum_{\sigma \in \Sigma} \ell_\sigma$.  Let $G_{\b \ell_\Sigma} (\Sigma)$ denote the pairing 
obtained from $\Sigma$ by expanding the bridge $\sigma$ into $\ell_\sigma$ parallel bridges, 
for each $\sigma \in \Sigma$.  Thus $\Gamma$ may be recovered from its skeleton through $\Gamma 
= G_{\b \ell_\Sigma}(\Sigma)$ for a unique family $\b \ell_\Sigma$.
For given $p \in \N$, the sum over all pairings $\Gamma$ satisfying $\abs{\Gamma} = p$ 
therefore becomes
\begin{equation} \label{estimate on R in terms of skeletons}
\sum_{n + n' = 2p}\; \sum_{\Gamma \in \scr P_{n,n'}^*} R_x(\Gamma)
\;=\; \sum_{m + m' \leq 2p} \;
\sum_{\Sigma \in \scr S^*_{m ,m'}} \; \sum_{\b \ell_\Sigma \,:\, \abs{\b \ell_\Sigma} = p} 
R_x(G_{\b \ell_\Sigma}(\Sigma))\,.
\end{equation}

Next, we define and estimate the contribution to $R_x(\Gamma)$ of a set of $\ell$ parallel 
bridges.  Let $\ell \geq 1$, and two labels $y,z$ be given. Then we define
\begin{equation*}
D_\ell(y,z) \;\deq\; \sum_{x_0, \dots, x_\ell} \delta_{x_0 y} \delta_{x_\ell z}
\prod_{i = 0}^{\ell - 1} \ind{1 \leq \abs{x_i - x_{i + 1}} \leq W}\,.
\end{equation*}
Thus, $D_\ell(y,z)$ is equal to the number of paths of length $\ell$ from $y$ to $z$, whereby 
each step takes values in $\{x \,:\, 1 \leq \abs{x} \leq W\}$. (We could also have included the 
nonbacktracking restriction in the definition of $D_\ell$, but this is not needed as we only 
want an upper bound on $R_x(\Sigma)$). Graphically, $D_\ell$ corresponds to the contribution of 
$\ell$ parallel bridges; see Figure \ref{figure: many parallel bridges}.
\begin{figure}[ht!]
\begin{center}
\includegraphics{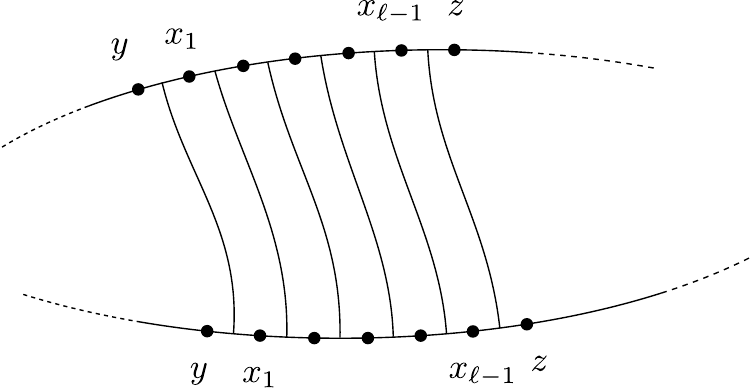}
\end{center}
\caption{Summing up $\ell$ parallel bridges. \label{figure: many parallel bridges}}
\end{figure}

We need the following straightforward properties of $D_\ell$.
\begin{lemma} \label{lemma: heat kernel bounds}
Let $\ell \in \N$. Then for each $y$ we have
\begin{equation*}
\sum_{z} D_\ell(y,z) \;=\; M^{\ell}\,.
\end{equation*}
Moreover, for each $y$ and $z$ we have
\begin{equation*}
D_\ell(y,z) \;\leq\; M^{\ell - 1}\,.
\end{equation*}
as well as
\begin{equation*}
D_\ell(y,z) \;\leq\; \frac{C}{\ell^{d/2}} M^{\ell - 1} + \frac{C}{N^d} M^{\ell}\,.
\end{equation*}
for some constant $C$.
\end{lemma}
\begin{proof}
The first two statements are obvious. The last follows from a standard local central limit 
theorem; see for instance the proof in \cite{So1}.
\end{proof}

\subsection{Orbits of vertices} \label{subsection: orbit of vertices}
Fix $\Gamma \in \scr P^*_{n,n'}$. We observe that the product in \eqref{expression for R for 
pairings} may be interpreted as an indicator function that fixes labels along paths of 
vertices. To this end, we define a map $\tau \equiv \tau_{\,\Gamma}$ on the vertex set $\cal 
V_{n,n'}$. Start with a vertex $i \in \cal V_{n,n'}$.  Let $e$ be the outgoing edge of $i$ 
(i.e.\ $e = (i,i+1)$), and $e'$ the edge bridged by $\Gamma$ to $e$.  Then we define $\tau i$ 
as the final vertex of $e'$ (i.e.\ $e'=(\tau i - 1, \tau i)$).  Thus the product in 
\eqref{expression for R for pairings} may be rewritten as
\begin{equation*}
\prod_{\{e, e'\} \in \Gamma} \ind{x_{a(e)}  = x_{b(e')}} \ind{x_{b(e)} =  x_{a(e')}} \;=\; 
\prod_{i \in \cal V_{n,n'}} \delta_{x_i x_{\tau i}}\,.
\end{equation*}
Starting from any vertex $i \in \cal V_{n,n'}$ we construct a path $(i, \tau i, \tau^2 
i,\dots)$.  In this fashion the set of vertices is partitioned into orbits of $\tau$; see 
Figure \ref{figure: path construction}. Let $[i] \subset \cal V_{n,n'}$ denote the orbit of the 
vertex $i \in \cal V_{n,n'}$.
\begin{figure}[ht!]
\begin{center}
\includegraphics{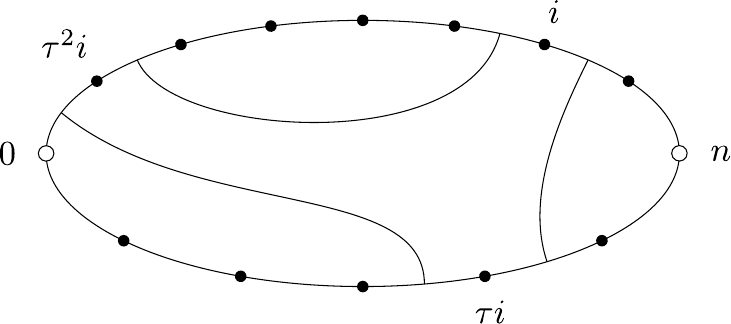}
\caption{Construction of the orbit $[i]$ of the vertex $i$. \label{figure: path construction}}
\end{center}
\end{figure}

Next, let $\Sigma = S(\Gamma) \in \scr S^*_{m,m'}$ be the skeleton pairing of $\Gamma$, and let the family $\b 
\ell_\Sigma$ be defined through $\Gamma = G_{\b \ell_\Sigma}(\Sigma)$. The map $\tau \equiv \tau_{\,\Sigma}$ on the 
skeleton pairing $\Sigma$ is defined exactly as for $\Gamma$ above. In order to sum over all labels $\b x = (x_0, \dots, 
x_{n+n'-1})$ in the expression for $R_x(G_{\b \ell_\Sigma}(\Sigma))$, we split the set of labels $\b x$ into two parts: 
labels of vertices between two parallel bridges, and labels associated with vertices of $\Sigma$. In order to make this 
precise, we need the following definitions.

Let $Z(\Sigma)$ be the set of orbits of $\Sigma$. It contains the distinguished orbits $[0]$ and $[m]$, which receive 
the labels $0$ and $x$ respectively. (Note that we may have $[0] = [m]$, in which case $x$ must be $0$.)  We assign a 
label $y_\zeta$ to each orbit $\zeta \in Z(\Sigma)$, and define the family $\b y_\Sigma \deq \{y_\zeta\}_{\zeta \in 
Z(\Sigma)}$.  Each bridge $\sigma \in \Sigma$ ``sits between two orbits'' $\zeta_1(\sigma)$ and $\zeta_2(\sigma)$ . More 
precisely, let $e = (i,i+1) \in \sigma$ be the smaller edge of $\sigma$. Then we set $\zeta_1(\sigma) \deq [i]$ and 
$\zeta_2(\sigma) \deq [i+1]$. (Note that using the larger edge of $\sigma$ in this definition would simply exchange 
$\zeta_1(\sigma)$ and $\sigma_2(\sigma)$; this is of no consequence for the following.)

\begin{lemma} \label{lemma: summing over orbit labels}
For given $\Sigma \in \scr S^*_{m,m'}$, $\b \ell_\Sigma$, and $\Gamma = G_{\b \ell_\Sigma}(\Sigma) \in \scr P^*_{n,n'}$ 
we have
\begin{equation} \label{bound in terms of lump indices}
R_x(\Gamma) \;\leq\; \frac{1}{(M - 1)^{\bar n}}
\sum_{\b y_{\Sigma}} \ind{0 = y_{[0]} } \ind{x = y_{[m]}} \prod_{\sigma \in \Sigma}
D_{\ell_\sigma}(y_{\zeta_1(\sigma)}, y_{\zeta_2(\sigma)})\,.
\end{equation}
\end{lemma}

\begin{proof}
The left-hand side of \eqref{bound in terms of lump indices} is given by the expression \eqref{expression for R for 
pairings}. The summation over all $x_i$'s between parallel bridges of $\Gamma$ is contained in the factors $D_\ell$, and 
the summation over all the remaining $x_i$'s is replaced by the sum over $\b y_{\Sigma}$. We relaxed the nonbacktracking 
condition in $Q_x(\b x)$ to obtain an upper bound.
\end{proof}

Next, let $Z^*(\Sigma) \deq Z(\Sigma) \setminus \{[0]\}$ and define $L(\Sigma) \deq \abs{Z^*(\Sigma)}$. The set 
$Z^*(\Sigma)$ is the set of orbits whose label is summed over in $\sum_x R_x(\Gamma)$. The following lemma gives an 
upper bound on $L(\Sigma)$. It states, roughly, that the number of orbits (or free labels) is bounded by $2 \bar m / 3$; 
we refer to it as the $2/3$ rule. Compare this bound with the trivial bound $L(\Sigma) \leq \bar m$, which would be 
sharp if $\Sigma$ were allowed to have parallel bridges.

\begin{lemma}[The $2/3$ rule] \label{lemma: bound on number of orbits}
Let $\Sigma \in \scr S_{m,m'}^*$. Then $L(\Sigma) \leq \frac{2 \bar m}{3} + \frac{1}{3}$.
\end{lemma}
\begin{proof}
Let $Z'(\Sigma) \deq Z(\Sigma) \setminus \{[0], [m]\}$.
We show that every orbit $\zeta \in Z'(\Sigma)$ consists of at least 3 vertices.
Let $i \in \cal V_{m,m'}$ belong to $\zeta \in Z'(\Sigma)$. Then, by Lemma \ref{lemma: properties of skeleton pairings} 
(ii), we have that $\tau i \neq i$.  By assumption, $\tau i \notin \{0, m\}$.  Hence $\tau^2 i \neq i$, for otherwise 
$\Sigma$ would have two parallel bridges, in contradiction to Lemma \ref{lemma: properties of skeleton pairings} (i).  
Therefore the orbit of $\tau$ contains at least 3 vertices.  Note that there are orbits containing exactly 3 vertices, 
as depicted in Figure \ref{figure: path construction}.

The total number of vertices of $\Sigma$ not including the vertices $0$ and $m$ is $2 \bar m - 
2$, so that we get
\begin{equation*}
3 \abs{Z'(\Sigma)} \;\leq\; 2 \bar m - 2\,.
\end{equation*}
The claim follows from the bound $\abs{Z^*(\Sigma)} \leq \abs{Z'(\Sigma)} + 1$.
\end{proof}

\subsection{Bound on $R_x(\Gamma)$}
As in the previous subsection, we fix $\Gamma \in \scr P^*_{n,n'}$, $\Sigma = S(\Gamma) \in 
\scr S^*_{m,m'}$, and $\b \ell_\Sigma$ satisfying $\Gamma = G_{\b \ell_\Sigma}(\Sigma)$.

We start by observing that the product in \eqref{bound in terms of lump indices} may be rewritten in terms of a 
multigraph $\Pi(\Sigma)$ on the vertex set $Z(\Sigma)$. Each factor $D_{\ell_\sigma}(y_{\zeta_1(\sigma)}, 
y_{\zeta_2(\sigma)})$ yields an edge connecting the orbits $\zeta_1$ and $\zeta_2$. In other words, there is a 
one-to-one map, which we denote by $\phi$, between bridges of $\Sigma$ and edges of $\Pi(\Sigma)$; each bridge $\sigma 
\in \Sigma$ gives rise to an edge $\phi(\sigma)$ of $\Pi(\Sigma)$ connecting $\zeta_1(\sigma)$ and $\zeta_2(\sigma)$.  
See Figure \ref{figure: multigraph on orbits} for an example of such a multigraph.
\begin{figure}[ht!]
\begin{center}
\includegraphics{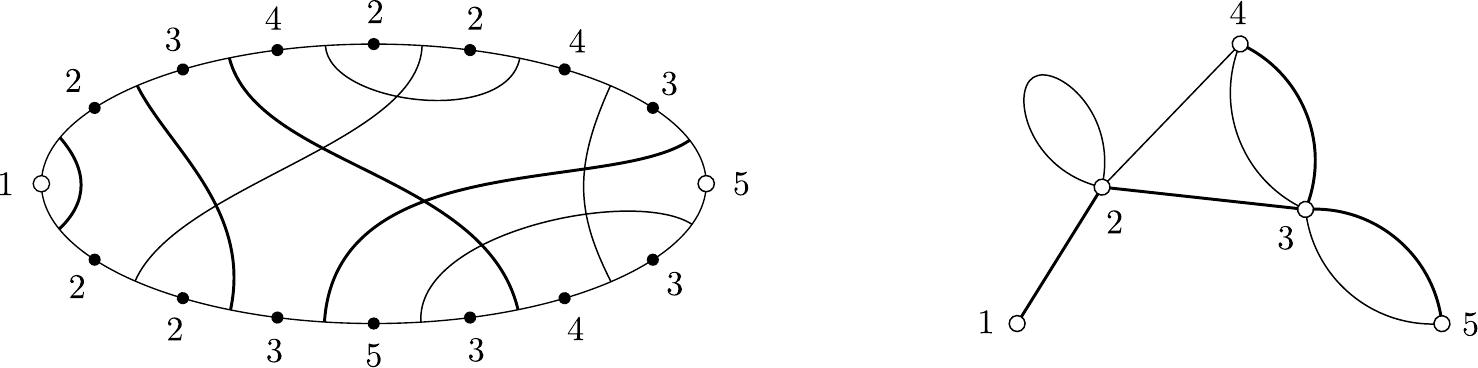}
\end{center}
\caption{Left: a skeleton pairing $\Sigma$ giving rise to $5$ orbits indexed by $Z(\Sigma) = \{1, \dots, 5\}$; the 
bridges in $\Sigma_T$, for one possible choice of $\Sigma_T$, are drawn using thick lines. Right: the corresponding 
multigraph on the vertex set $Z(\Sigma)$; the edges in $\phi(\Sigma_T)$ are drawn using thick lines.
\label{figure: multigraph on orbits}}
\end{figure}

\begin{lemma}
There is a subset of bridges $\Sigma_T \subset \Sigma$ of size $\abs{\Sigma_T} = L(\Sigma)$, such that, in the subgraph 
of $\Pi(\Sigma)$ with the edge set $\phi(\Sigma_T)$, each orbit $\zeta \in Z^*(\Sigma)$ is connected to $[0]$.
\end{lemma}
\begin{proof}
Starting from $\zeta_0 = [0]$, we construct a sequence of orbits $\zeta_0, \zeta_1, \dots, \zeta_{L(\Sigma)}$, and a 
sequence of bridges $\sigma_1, \dots, \sigma_{L(\Sigma)}$, with the property that for all $k = 1, \dots, L(\Sigma)$ 
there is a $k' < k$ such that $\zeta_k$ and $\zeta_{k'}$ are connected by $\phi(\sigma_k)$.

Assume that $\zeta_0, \dots, \zeta_{k - 1}$ have already been constructed.  Let $i$ be the smallest vertex of $\cal 
V_{m,m'} \setminus (\zeta_0 \cup \cdots \cup \zeta_{k - 1})$.  Then we set $\zeta_k = [i]$.  By construction, the vertex 
$i - 1$ belongs to an orbit $\zeta_{k'}$ for some $k' < k$. Set $\sigma_k$ to be the bridge containing $\{i-1, i\}$. 
Hence, by definition of $\Pi(\Sigma)$, we see that $\zeta_k$ and $\zeta_{k'}$ are connected by $\phi(\sigma_k)$.

The set $\Sigma_T$ is given by $\{\sigma_1, \dots, \sigma_{L(\Sigma)}\}$.
\end{proof}

Because $\abs{\Sigma_T} = L(\Sigma)$, the subgraph of $\Pi(\Sigma)$ with the edge set
$\phi(\Sigma_T)$ is a tree that connects all orbits in $Z^*(\Sigma)$ to $[0]$. Let us call this tree $\cal T(\Sigma)$. 
Its root is $[0]$.

Next, we observe that
\begin{equation} \label{lower bound on number of available bridges}
\abs{\Sigma \setminus \Sigma_T} \;\geq\; 1\,.
\end{equation}
Indeed, using Lemma \ref{lemma: bound on number of orbits} and $\bar m \geq 2$ we find
\begin{equation} \label{proof of lower bound on number of available bridges}
\abs{\Sigma \setminus \Sigma_T} \;=\; \bar m - L(\Sigma) \;\geq\; \frac{\bar m}{3} - \frac{1}{3}\;\geq\; \frac{1}{3}\,.
\end{equation}

We now estimate \eqref{bound in terms of lump indices} as follows. Each factor indexed by $\sigma \in \Sigma \setminus 
\Sigma_T$ is estimated by $\sup_{y,z} D_{\ell_\sigma}(y,z)$. As it turns out, we need to exploit the heat kernel decay 
for at least one bridge in $\Sigma \setminus \Sigma_T$. Pick a bridge $\bar{\sigma} \in \Sigma \setminus \Sigma_T$ (By 
\eqref{lower bound on number of available bridges} there is such a bridge).  Using Lemma \ref{lemma: heat kernel 
bounds}, we estimate
\begin{subequations}
\label{linfty-bound}
\begin{align} \label{linfty-bound1}
\sup_{y,z} D_{\ell_\sigma}(y,z) &\;\leq\; M^{\ell_\sigma - 1} & &\text{if } \sigma \in \Sigma \setminus (\Sigma_T \cup 
\{\bar{\sigma}\})\,,
\\ \label{linfty-bound2}
\sup_{y,z} D_{\ell_\sigma}(y,z) &\;\leq\; \frac{C}{\ell_\sigma^{d/2}} M^{\ell_\sigma - 1} + \frac{C}{N^d} 
M^{\ell_\sigma} & &\text{if } \sigma  = \bar{\sigma}\,.
\end{align}
\end{subequations}
Since $N \geq W M^{1/6}$ and $M \sim C W^d$ we find
\begin{equation*}
\frac{C}{\ell_{\bar{\sigma}}^{d/2}} M^{\ell_{\bar{\sigma}} - 1} + \frac{C}{N^d} M^{\ell_{\bar{\sigma}}} \;\leq\; C 
M^{\ell_{\bar{\sigma}} - 1} \pbb{\frac{1}{\ell_{\bar{\sigma}}^{1/2}} + \frac{1}{M^{1/6}}}\,
\end{equation*}
where we replaced $d$ with $1$ to obtain an upper bound.
Thus we get
\begin{equation*}
\sum_x R_x(\Gamma) \;\leq\; \frac{C}{(M - 1)^{\bar n}} \pbb{\frac{1}{\ell_{\bar{\sigma}}^{1/2}} + \frac{1}{M^{1/6}}} 
\prod_{\sigma \in \Sigma \setminus \Sigma_T}
M^{\ell_\sigma - 1}
\sum_{\b y_{\Sigma}} \ind{0 = y_{[0]} }  \prod_{\sigma \in \Sigma_T}
D_{\ell_\sigma}(y_{\zeta_1(\sigma)}, y_{\zeta_2(\sigma)})\,.
\end{equation*}
We perform the summation over $\b y_\Sigma$ by starting at the leaves of $\cal T(\Sigma)$ and moving towards the root 
$[0]$. Each vertex $\zeta$ of $\cal T(\Sigma)$ carries a label $y_\zeta$.  Let us choose a leaf $\zeta$ of $\cal 
T(\Sigma)$, and denote by $\zeta'$ the parent of $\zeta$ in $\cal T(\Sigma)$.  Let $\sigma \in \Sigma$ be the (unique) 
bridge such that $\phi(\sigma)$ connects $\zeta$ and $\zeta'$. Then summation over $y_\zeta$ yields the factor
\begin{equation} \label{l1-bound}
\sum_{y_\zeta} D_{\ell_\sigma}(y_{\zeta}, y_{\zeta'}) \;=\; M^{\ell_\sigma}\,,
\end{equation}
by Lemma \ref{lemma: heat kernel bounds}. Continuing in this manner until we reach the root, we find
\begin{align*}
\sum_x R_x(\Gamma) &\;\leq\; \frac{C}{(M - 1)^{\bar n}} \pbb{\frac{1}{\ell_{\bar{\sigma}}^{1/2}} + \frac{1}{M^{1/6}}} 
\prod_{\sigma \in \Sigma \setminus \Sigma_T}
M^{\ell_\sigma - 1}
\prod_{\sigma \in \Sigma_T} M^{\ell_\sigma}
\\
&\;=\;
C \, \pbb{\frac{M}{M - 1}}^{\bar n} \pbb{\frac{1}{\ell_{\bar{\sigma}}^{1/2}} + \frac{1}{M^{1/6}}} 
\frac{1}{M^{\abs{\Sigma \setminus \Sigma_T}}}\,.
\end{align*}
Now \eqref{proof of lower bound on number of available bridges} implies
\begin{equation*}
\abs{\Sigma \setminus \Sigma_T} \;\geq\;
\frac{\bar m}{3} - \frac{1}{3}\,,
\end{equation*}
so that
\begin{equation}\label{estimate on R(Gamma)}
\sum_x R_x(\Gamma) \;\leq\; C \pbb{\frac{M}{M - 1}}^{\bar n} \pbb{\frac{1}{\ell_{\bar{\sigma}}^{1/2}} + 
\frac{1}{M^{1/6}}} \frac{
M^{1/3}}{M^{\bar m / 3}}\,.
\end{equation}

Notice that \eqref{estimate on R(Gamma)} results from an $\ell^1$-$\ell^\infty$-summation procedure, where the 
$\ell^1$-bound \eqref{l1-bound} was used for propagators associated with bridges in $\Sigma_T$, and the 
$\ell^\infty$-bound \eqref{linfty-bound} for propagators associated with bridges in $\Sigma \setminus \Sigma_T$. The 
bound \eqref{linfty-bound1} is a simple power counting bound; the bound \eqref{linfty-bound2}, improved by the heat 
kernel decay, is used only for one bridge. Note that in the original setup \eqref{hhormalization} each row and column of 
$H$ contains $M$ nonzero entries $H_{xy}$, whose positions are determined by the condition $1 \leq \abs{x - y} \leq W$.  
If we removed this last condition and only required that each row and column contain $M$ nonzero entries in arbitrary 
locations off the diagonal, then all bounds relying solely on power counting would remain valid. In particular, 
\eqref{estimate on R(Gamma)} would be valid without the factor $\ell_{\bar{\sigma}}^{-1/2} + M^{-1/6}$, which results 
from the heat kernel decay associated with the special band structure.


\subsection{Sum over pairings}
We may now estimate $\sum_{n + n' = 2p} \sum_{\Gamma \in \scr P_{n,n'}^*} \sum_x R_x(\Gamma)$ 
for fixed $p$. Let first $p,m,m' \geq 0$ and $\Sigma \in \scr S^*_{m,m'}$. Then \eqref{estimate 
on R(Gamma)} yields
\begin{equation*}
\sum_{\b \ell_\Sigma \,:\, \abs{\b \ell_\Sigma} = p} \sum_x R_x(G_{\b \ell_{\Sigma}}(\Sigma))
\;\leq\; C \pbb{\frac{M}{M - 1}}^p \frac{M^{1/3}}{M^{\bar m / 3}} \sum_{\b \ell_\Sigma \,:\, \abs{\b \ell_\Sigma} = p} 
\pbb{\frac{1}{\ell_{\bar{\sigma}}^{1/2}} + \frac{1}{M^{1/6}}}\,.
\end{equation*}
The sum on the right-hand side is equal to
\begin{align*}
\sum_{\ell_1 + \cdots + \ell_{\bar m} = p} \pbb{\frac{1}{\ell_1^{1/2}} + \frac{1}{M^{1/6}}}
&\;=\; \sum_{\ell_1 = 1}^{p - \bar m + 1} \pbb{\frac{1}{\ell_1^{1/2}} + \frac{1}{M^{1/6}}}
\sum_{\ell_2 + \cdots + \ell_{\bar m} = p - \ell_1} 1
\\
&\;\leq\; \sum_{\ell_1 = 1}^p \pbb{\frac{1}{\ell_1^{1/2}} + \frac{1}{M^{1/6}}}
\binom{p - \ell_1 - 1}{\bar m - 2}
\\
&\;\leq\; C \pbb{\frac{1}{p^{1/2}} + \frac{1}{M^{1/6}}} \frac{p^{\bar m - 1}}{(\bar m - 2)!}\,.
\end{align*}

Next, we note that
\begin{equation*}
\abs{\scr S^*_{m,m'}} \;\leq\; (2 \bar m - 1)(2 \bar m - 3) \cdots 3 \cdot 1 \;\leq\; 2^{\bar m} \, \bar m\, !\,.
\end{equation*}
This expresses the fact that the first edge of $\Sigma$ can be bridged with at most $(2 \bar m - 1)$ edges, the next 
remaining edge with at most $(2 \bar m - 3)$ edges, and so on. Therefore \eqref{estimate on R in terms of skeletons} and 
Lemma \ref{lemma: properties of skeleton pairings} (iii) yield
\begin{align*}
\sum_{n + n' = 2p}\; \sum_{\Gamma \in \scr P_{n,n'}^*} \sum_x R_x(\Gamma)
&\;\leq\; C \sum_{4 \leq m + m' \leq 2p}  2^{\bar m} \bar m!
\pbb{\frac{M}{M - 1}}^p \frac{M^{1/3}}{M^{\bar m / 3}} \frac{p^{\bar m - 1}}{(\bar m - 2)!} \pbb{\frac{1}{p^{1/2}} + 
\frac{1}{M^{1/6}}}
\\
&\;\leq\; C \frac{M^{1/3}}{p} \pbb{\frac{1}{p^{1/2}} + \frac{1}{M^{1/6}}} \pbb{\frac{M}{M - 1}}^p \sum_{4 \leq m + m' 
\leq 2p} \bar m^2 2^{\bar m} \pbb{\frac{p}{M^{1/3}}}^{\bar m}
\\
&\;\leq\; \frac{M^{1/3}}{p} \pbb{\frac{1}{p^{1/2}} + \frac{1}{M^{1/6}}} \pbb{\frac{M}{M - 1}}^p \sum_{\bar m = 2}^p 
\pbb{\frac{Cp}{M^{1/3}}}^{\bar m}\,.
\end{align*}
Thus, Lemma \ref{lemma: sum over non-pairings} yields
\begin{equation} \label{bound on sum of graphs}
\sum_{n + n' = 2p} h_{n,n'} \;\leq\; \frac{M^{1/3}}{p} \pbb{\frac{1}{p^{1/2}} + \frac{1}{M^{1/6}}} \pbb{\frac{M}{M - 
1}}^p \sum_{r = 2}^p \pbb{\frac{Cp}{M^{1/3}}}^r\,,
\end{equation}
where we abbreviated
\begin{equation*}
h_{n,n'} \;\deq\; \sum_{\Gamma \in \scr G^*_{n,n'}} \sum_x V_x(\Gamma)\,.
\end{equation*}

\subsection{Conclusion of the proof} \label{subsection: cooclusion of the proof}
In this subsection we complete the proof of Proposition \ref{proposition: main estimate on error} by showing that the 
error
\begin{equation} \label{full error term}
E_W \;\deq\; \sum_{n,n'} \absb{a_n(\eta T) \, a_{n'}(\eta T)} \, h_{n,n'}
\end{equation}
satisfies $E_W = o(1)$ as $W \to \infty$, uniformly in $N \geq W^{1 + d/6}$.

We begin by deriving bounds on the coefficients $a_n(t)$.
\begin{lemma} \label{lemma: bounds on coefficients a_n}
\begin{enumerate}
\item
We have
\begin{equation} \label{sum of coefficients a_n}
\sum_{n \geq 0} \abs{a_n(t)}^2 \;=\; 1 + O(M^{-1})\,,
\end{equation}
uniformly in $t \in \R$.
\item
We have
\begin{equation} \label{bound on coefficients a_n}
\abs{a_n(t)} \;\leq\; C \frac{t^n}{n!}\,.
\end{equation}
\end{enumerate}
\end{lemma}

\begin{proof}
We start with (i). Write
\begin{equation*}
\sum_{n \geq 0} \abs{a_n(t)}^2 \;=\;
\sum_{n \geq 0} \sum_{k, k' \geq 0} \frac{\alpha_{n + 2k}(t) \, \ol{\alpha_{n + 2k'}(t)}}{(M - 
1)^{k + k'}}\,.
\end{equation*}
The term $k = k' = 0$ yields $1$ by \eqref{orthonormality of Chebyshev coefficients}. The rest 
is equal, by \eqref{orthonormality of Chebyshev coefficients}, to
\begin{equation*}
\sum_{n \geq 0} \sum_{k + k' > 0} \frac{\alpha_{n + 2k}(t) \, \ol{\alpha_{n + 2k'}(t)}}{(M - 
1)^{k + k'}}
\;\leq\;
\sum_{k + k' > 0} \sum_{n \geq 0} \frac{\abs{\alpha_{n + 2k}(t)}^2}{(M - 1)^{k + k'}}
\;\leq\;
\sum_{k + k' > 0} \frac{1}{(M - 1)^{k + k'}} \;=\; O(M^{-1})\,.
\end{equation*}

In order to prove (ii), we use the integral representation (see \cite{GradshteynRyzhik})
\begin{equation*}
J_n(t) \;=\; \frac{\pb{\frac{t}{2}}^n}{\sqrt{\pi} \, \Gamma\pb{n + \frac{1}{2}}} \int_{-1}^1 
\dd \lambda \; \ee^{\ii t \lambda} (1 - \lambda^2)^{n - \frac{1}{2}}\,.
\end{equation*}
Therefore
\begin{equation} \label{bound on alpha}
\abs{\alpha_n(t)} \;\leq\; 2 \frac{n + 1}{t} \, \frac{\pb{\frac{t}{2}}^{n+1}}{\sqrt{\pi} \, 
\Gamma\pb{n + \frac{3}{2}}} \, \frac{\pi}{2} \;\leq\; \frac{t^n}{n!}\,.
\end{equation}
Moreover, \eqref{orthonormality of Chebyshev coefficients} yields
\begin{equation} \label{stupid bound on alpha}
\abs{\alpha_n(t)} \;\leq\; 1\,.
\end{equation}
We use the estimate
\begin{equation*}
\abs{a_n(t)} \;\leq\; \sum_{k \geq 0} \frac{\abs{\alpha_{n + 2k}(t)}}{(M - 1)^k}\,.
\end{equation*}
Let us first consider the case $t \leq n$. Then it is easy to see that
$\frac{t^{n + 2k}}{(n + 2k)!} \leq \frac{t^n}{n!}$.
Together with \eqref{bound on alpha} this yields
\begin{equation*}
\abs{a_n(t)} \;\leq\; \sum_{k \geq 0} \frac{1}{(M - 1)^k} \frac{t^{n + 2k}}{(n + 2k)!}
\;\leq\; \frac{t^n}{n!} \sum_{k \geq 0} \frac{1}{(M - 1)^k}
\;\leq\; C \frac{t^n}{n!}\,.
\end{equation*}

If $t > n$ we have $\frac{t^n}{n!} \geq C$.
Thus the bound \eqref{stupid bound on alpha} yields
\begin{equation*}
\abs{a_n(t)} \;\leq\; \sum_{k \geq 0} \frac{C}{(M - 1)^k} \;\leq\; C \, \frac{t^n}{n!}\,.
\qedhere
\end{equation*}
\end{proof}

Using the new variables $p \deq \bar n = \frac{n + n'}{2}$ and $q \deq \frac{n - n'}{2}$ we 
find from the definition \eqref{full error term}
\begin{equation*}
E_W \;\leq\; \sum_{p \geq 0} \sum_{q = -p}^p \abs{a_{p+q}(\eta T) a_{p - q}(\eta T)} \, h_{p+q, 
p-q}\,.
\end{equation*}

Next, we observe that Lemma \ref{lemma: bounds on coefficients a_n} (ii) implies that terms 
corresponding to $n,n' \gg t = \eta T \sim C M^\kappa T$ are strongly suppressed. Thus we 
introduce a cutoff at $p = M^\mu$, where $\kappa < \mu < \frac{1}{3}$.
Let us first consider the terms $p \leq M^\mu$. We need to estimate
\begin{align*}
E_W^{\leq} \;&\deq\; \sum_{p = 0}^{M^\mu} \sum_{q = -p}^p \abs{a_{p+q}(\eta T) a_{p - q}(\eta T)} \, 
h_{p+q, p-q}
\\
&\;\leq\; \pBB{\sum_{p = 0}^{M^\mu} \sum_{q = -p}^p \absb{a_{p+q}(\eta T) a_{p - q}( \eta 
T)}^2}^{1/2}
\pBB{\sum_{p = 0}^{M^\mu} \sum_{q = -p}^p \pb{h_{p+q, p-q}}^2
}^{1/2}
\\
&\;\leq\; C \pBB{\sum_{p = 0}^{M^\mu} \sum_{q = -p}^p \p{h_{p+q, p-q}}^2
}^{1/2}\,,
\end{align*}
where we used Lemma \ref{lemma: bounds on coefficients a_n} (i). Thus,
\begin{equation*}
\pb{E^\leq_W}^2 \;\leq\; C \sum_{p = 0}^{M^\mu} \qBB{\sum_{q = -p}^p h_{p+q, p-q}}^2
\;\leq\; \sum_{p = 2}^{M^\mu} \qBB{\frac{M^{1/3}}{p} \pbb{\frac{1}{p^{1/2}} + \frac{1}{M^{1/6}}} \pbb{\frac{M}{M - 1}}^p 
\sum_{\bar m = 2}^p \pbb{\frac{Cp}{M^{1/3}}}^{\bar m}}^2\,,
\end{equation*}
by \eqref{bound on sum of graphs}. For $p \leq M^\mu$ and $W$ large enough, the term in the 
square brackets is bounded by
\begin{equation*}
C \frac{M^{1/3}}{p} \pbb{\frac{1}{p^{1/2}} + \frac{1}{M^{1/6}}} \frac{1}{\pb{1 - \frac{1}{M}}^p} 
\pbb{\frac{p}{M^{1/3}}}^2
\;\leq\; C \, \frac{p^{1/2}}{M^{1/3}}
\;\leq\; \frac{C}{M^{1/6}}\,.
\end{equation*}
Thus we find $\pb{E^\leq_W}^2 \leq C M^{\mu - 1/3}$.

Let us now consider the case $p > M^\mu$, i.e.\ estimate
\begin{equation*}
E_W^{>} \;\deq\; \sum_{p > M^\mu} \sum_{q = -p}^p \absb{a_{p+q}(\eta T) a_{p - q}(\eta T)} \, h_{p+q, 
p-q}\,.
\end{equation*}
By \eqref{bound on coefficients a_n} and the elementary inequality $\frac{p!}{(p-q)!} \leq 
\frac{(p+q)!}{p!}$ we have
\begin{equation*}
\abs{a_{p+q}(t) a_{p-q}(t)} \;\leq\; C \frac{t^{2p}}{(p+q)! (p-q)!} \;\leq\; C \frac{t^{2p}}{p!  
p!}\,.
\end{equation*}
This gives
\begin{align*}
E_W^{>} &\;\leq\; C \sum_{p > M^\mu} \frac{(\eta T)^{2p}}{p! p!} \sum_{q = -p}^p h_{p+q, p-q}
\\
&\;\leq\; C \sum_{p > M^\mu} \frac{(\eta T)^{2 p}}{p! p!} \frac{M^{1/3}}{p} \pbb{\frac{1}{p^{1/2}} + \frac{1}{M^{1/6}}} 
\pbb{\frac{M}{M - 1}}^p \sum_{\bar m = 2}^p \pbb{\frac{Cp}{M^{1/3}}}^{\bar m}
\end{align*}
by \eqref{bound on sum of graphs}. Setting $\eta\sim C M^\kappa$ yields
\begin{align*}
E_W^{>} &\;\leq\; \sum_{p > M^\mu} \frac{(C M^\kappa T)^{2p}}{p! p!} \sum_{\bar m = 0}^{p-2} 
\pbb{\frac{Cp}{M^{1/3}}}^{\bar m}
\\
&\;\leq\; \sum_{p > M^\mu} \pbb{\frac{C M^\kappa T}{p}}^{2p} + \sum_{p > M^\mu} \pbb{\frac{C 
M^{2\kappa} T^2}{p M^{1/3}}}^p
\\
&\;\leq\; \sum_{p > M^\mu} \pb{C M^{\kappa - \mu} T}^{2p} + \sum_{p > M^\mu} \pb{C M^{2 \kappa 
- 1/3 - \mu} T^2}^p
\\
&\;\leq\; \pb{C M^{\kappa - \mu} T}^{2 M^\mu} +  \pb{C M^{2 \kappa - 1/3 - \mu} T}^{M^\mu}\,.
\end{align*}
Choosing $\mu = 1/3 - \beta$ (where, we recall, $0 < \beta < 2/3 - 2 \kappa$) completes the 
proof of Proposition \ref{proposition: main estimate on error}.

\section{The ladder pairings} \label{section: ladder}

In this section we analyse the contribution of the ladder pairings, $\sum_{n \geq 0} \abs{a_n(\eta T)}^2 \, V_x(L_n)$,
and complete the proof of Theorem \ref{theorem: first main result}. (Recall that $\eta \deq W^{d \kappa}$ is the time 
scale.)
Recalling the expression \eqref{expression for V for pairings}, and noting that in the case of 
the ladder the variables $x_0, \dots, x_n$ determine the value of all variables $x_0, \dots, 
x_{2n-1}$, we readily find
\begin{multline} \label{value of ladder}
V_x(L_n) \;=\; \frac{1}{(M - 1)^n} \sum_{\b x \in \Lambda_N^{n+1}} \delta_{0 x_0} \delta_{x 
x_n} \prod_{i = 0}^{n - 1} \ind{1 \leq \abs{x_{i+1} - x_i} \leq W}
\\
\times
\prod_{i = 0}^{n - 2} \ind{x_i \neq x_{i+2}} \prod_{0 \leq i < j \leq n-1} \indb{\{x_i, 
x_{i+1}\} \neq \{x_j, x_{j+1}\}}\,.
\end{multline}
Throughout this section we assume that $\eta = W^{d \kappa}$ for some $\kappa < 1/3$.

We perform a series of steps to simplify the expression \eqref{value of ladder}.
In a first step, we get rid of the last product.

\begin{lemma} \label{lemma: first step for ladders}
Under the assumptions of Proposition \ref{proposition: main estimate on error} we have
\begin{equation*}
\sum_{n \geq 0} \abs{a_n(\eta T)}^2 \, V_x(L_n) \;=\; \sum_{n \geq 0} \abs{a_n(\eta T)}^2 V^1_x(n) - 
E^1_x\,,
\end{equation*}
where
\begin{equation*}
V^1_x(n) \;\deq\; \frac{1}{(M - 1)^n} \sum_{\b x \in \Lambda_N^{n+1}} \delta_{0 x_0} \delta_{x 
x_n} \prod_{i = 0}^{n - 1} \ind{1 \leq \abs{x_{i+1} - x_i} \leq W}
\prod_{i = 0}^{n - 2} \ind{x_i \neq x_{i+2}}
\end{equation*}
and
\begin{equation*}
\sum_x \abs{E^1_x} \;\leq\; \frac{C}{W^{d \beta}}\,.
\end{equation*}
\end{lemma}

\begin{proof}
For each $\b x  = (x_0, \dots, x_n) \in \Lambda_N^{n+1}$ we write
$1 = \sum_{P \,:\, \{0, \dots, n-1\}} \Delta_P(\b x)$,
where the sum ranges ranges over all partitions $P$ of the set $\{0, \dots, n-1\}$, and $\Delta_P(\b x)$ is the 
indicator function
\begin{equation*}
\Delta_P(\b x) \;\deq\; \prod_{0 \leq i < j \leq n - 1} \begin{cases}
\indb{\{x_i, x_{i+1}\} = \{x_j, x_{j+1}\}} & \text{if $i$ and $j$ belong to the same lump of 
$P$}
\\
\indb{\{x_i, x_{i+1}\} \neq \{x_j, x_{j+1}\}} & \text{if $i$ and $j$ belong to different lumps 
of $P$}\,.
\end{cases}
\end{equation*}
Notice that if $P = P_0 \deq \hb{\{0\}, \dots, \{n-1\}}$ then $\Delta_P(\b x)$ is the last 
product of \eqref{value of ladder}. Let us define
\begin{equation*}
E_x^1(n) \;\deq\; \frac{1}{(M - 1)^n} \sum_{\b x \in \Lambda_N^{n+1}} \delta_{0 x_0} \delta_{x 
x_n} \prod_{i = 0}^{n - 1} \ind{1 \leq \abs{x_{i+1} - x_i} \leq W}
\prod_{i = 0}^{n - 2} \ind{x_i \neq x_{i+2}} \sum_{P \neq P_0} \Delta_P(\b x)\,.
\end{equation*}
Thus, by definition, we have
\begin{equation*}
E_x^1 \;=\; \sum_{n \geq 0} \abs{a_n(\eta T)}^2 E_x^1(n)\,.
\end{equation*}

Next, we estimate $\sum_x \abs{E_x^1}$. We begin by observing that each partition $P$ of $\{0, 
\dots, n-1\}$ uniquely defines a partition $\Gamma(P) \in \scr G_{n,n}^*$. Indeed, each lump $p 
\in P$ gives rise to the lump $\gamma \in \Gamma(P)$ defined by
$\gamma = \bigcup_{i \in p} \{e_i, e_{2n - 1 - i}\}$.
In particular, $\Gamma(P) \neq \Gamma(P')$ if $P \neq P'$.
We now claim that
\begin{equation*}
\frac{1}{(M - 1)^n} \sum_{\b x \in \Lambda_N^{n+1}} \delta_{0 x_0} \delta_{x x_n} \prod_{i = 
0}^{n - 1} \ind{1 \leq \abs{x_{i+1} - x_i} \leq W}
\prod_{i = 0}^{n - 2} \ind{x_i \neq x_{i+2}} \, \Delta_P(\b x)
\;\leq\; V_x(\Gamma(P))\,.
\end{equation*}
This can be directly read off \eqref{expression for V(Gamma)}; there is in fact an overcounting 
arising from the summation over $\b \pi_\Gamma$. Thus we find\begin{equation*}
\sum_x \abs{E^1_x} \;\leq\; \sum_x \sum_{n \geq 0} \abs{a_n(\eta T)}^2 \sum_{\substack{P \,:\, \{0, 
\dots, n-1\} \\ P \neq P_0}} V_x(\Gamma(P)) \;\leq\;
\sum_x \sum_{n \geq 0} \abs{a_n(\eta T)}^2 \sum_{\Gamma \in \scr G_{n,n}^*} V_x(\Gamma)\,.
\end{equation*}
Invoking Proposition \ref{proposition: main estimate on error} completes the proof.  
\end{proof}

In a second step, we get rid of the second to last product in \eqref{value of ladder}, i.e.\ 
the nonbacktracking condition.

\begin{lemma} \label{lemma: second step for ladders}
For any $T \geq 0$ we have
\begin{equation*}
\sum_{n \geq 0} \abs{a_n(\eta T)}^2 V_x^1(n)
\;=\;
\sum_{n \geq 0} \abs{a_n(\eta T)}^2 V_x^2(n) - E^2_x\,,
\end{equation*}
where
\begin{equation*}
V^2_x(n) \;\deq\; \frac{1}{(M - 1)^n} \sum_{\b x \in \Lambda_N^{n+1}} \delta_{0 x_0} \delta_{x 
x_n} \prod_{i = 0}^{n - 1} \ind{1 \leq \abs{x_{i+1} - x_i} \leq W}
\end{equation*}
and
\begin{equation*}
\sum_x \abs{E^2_x} \;\leq\; \frac{C}{W^{2d/3}}\,.
\end{equation*}
\end{lemma}

\begin{proof}
We find
\begin{equation*}
\sum_{x} \abs{E_x^2} \;=\; \sum_{n \geq 0} \abs{a_n(\eta T)}^2 \frac{1}{(M - 1)^n} \sum_{\b x \in 
\Lambda_N^{n+1}} \delta_{0 x_0}
\prod_{i = 0}^{n - 1} \ind{1 \leq \abs{x_{i+1} - x_i} \leq W} \qBB{1 - \prod_{i = 0}^{n - 2} 
\ind{x_i \neq x_{i+2}}}\,.
\end{equation*}
The expression in the square brackets is equal to
\begin{equation*}
1 - \prod_{i = 0}^{n - 2} \pb{1 - \ind{x_i = x_{i + 2}}} \;=\; \sum_{k = 1}^{n - 2} (-1)^{k + 
1} \sum_{0 \leq i_1 < \cdots < i_k \leq n - 2} \prod_{j = 1}^k \ind{x_{i_j} = x_{i_j + 2}}\,.
\end{equation*}
Therefore summing over $\b x$ yields
\begin{align*}
\sum_{x} \abs{E_x^2} &\;\leq\;
\sum_{n \geq 0} \abs{a_n(\eta T)}^2 \frac{1}{(M - 1)^n} \sum_{k = 1}^{n - 2} \binom{n - 2}{k} M^{n 
- k}
\\
&\;=\; \sum_{n \geq 0} \abs{a_n(\eta T)}^2 \pbb{\frac{M}{M - 1}}^n \qbb{\pbb{1 + \frac{1}{M}}^{n - 
2} - 1}\,.
\end{align*}
We introduce a cutoff at $n = M^{1/3}$. The part $n \leq M^{1/3}$ is bounded by
\begin{equation*}
\sum_{n \geq 0} \abs{a_n(\eta T)}^2 \pbb{\frac{M}{M - 1}}^M \qbb{\pbb{1 + \frac{1}{M}}^{M^{1/3}} - 
1} \;\leq\; C \pb{\ee^{M^{-2/3}} - 1} \;\leq\; \frac{C}{M^{2/3}}\,,
\end{equation*}
by Lemma \ref{lemma: bounds on coefficients a_n} (i).
The part $n > M^{1/3}$ is estimated using Lemma \ref{lemma: bounds on coefficients a_n} (ii), 
exactly as in the estimate of $E_W^{>}$ in Section \ref{subsection: cooclusion of the proof}.
\end{proof}

We summarize what we have proved so far.
\begin{lemma} \label{replacing as with alphas}
Under the assumptions of Proposition \ref{proposition: main estimate on error} we have
\begin{equation*}
\sum_{n \geq 0} \abs{a_n(\eta T)}^2 \, V_x(L_n) \;=\; \sum_{n \geq 0} \abs{\alpha_n(\eta T)}^2 P_x(n) + E_x\,,
\end{equation*}
where
\begin{equation*}
P_x(n) \;\deq\;  \frac{1}{M^n} \sum_{\b x \in \Lambda_N^{n+1}} \delta_{0 x_0} \delta_{x x_n} \prod_{i = 0}^{n - 1} 
\ind{1 \leq \abs{x_{i+1} - x_i} \leq W}
\end{equation*}
and
\begin{equation*}
\sum_{x} \abs{E_x} \;\leq\; \frac{C}{W^{d\beta}}\,.
\end{equation*}
\end{lemma}
\begin{proof}
The claim follows from Lemmas \ref{lemma: first step for ladders} and \ref{lemma: second step for ladders}, combined 
with an argument identical to the proof of Lemma \ref{lemma: bounds on coefficients a_n} (i) that allows us to replace 
$\abs{a_n(t)}^2$ with $\abs{\alpha_n(t)}^2$. We replaced the factor $\frac{1}{(M - 1)^n}$ with $\frac{1}{M^n}$ by 
introducing a cutoff at $n = M^{1/3}$, exactly as in the proof of Lemma \ref{lemma: second step for ladders}.
\end{proof}
The expression $P_x(n)$ is the (normalized) number of paths in $\Z^d$ of length $n$ from $0$ to any point in the set 
$x + N \Z^d$, whereby each step takes values in $\{y \,:\, 1 \leq \abs{y} \leq W\}$.

In a third step, we use the central limit theorem to replace $P_x(n)$ with a Gaussian.  Recall 
the definition of the heat kernel
\begin{equation*}
G(T, X) \;=\; \pbb{\frac{d + 2}{2 \pi T}}^{d/2} \, \ee^{- \frac{d+2}{2 T} \, \abs{X}^2}\,.
\end{equation*}

\begin{lemma} \label{lemma: central limit argument}
Let $\varphi \in C_b(\R^d)$ and $T \geq 0$.
Then we have
\begin{equation} \label{application of the CLT}
\lim_{W \to \infty} \sum_x P_x([\eta T]) \, \varphi \pbb{\frac{x}{W^{1 + d \kappa / 2}}}
\;=\; \int \dd X \; G(T, X) \varphi(X)\,,
\end{equation}
where $[\cdot]$ denotes the integer part.
\end{lemma}
\begin{proof}
Let $\widetilde{P}_x(n)$ denote the normalized number of paths in $\Z^d$ of length $n$ from $0$ 
to $x$, whereby each step takes values in $\{y \,:\, 1 \leq \abs{y} \leq W\}$. Then we have
\begin{align*}
\sum_{x \in \Lambda_N} P_x([\eta T]) \, \varphi \pbb{\frac{x}{W^{1 + d \kappa/2}}}
&\;=\; \sum_{x \in \Lambda_N} \sum_{\nu \in \Z^d} \widetilde{P}_{x + \nu N}([\eta T]) \, \varphi 
\pbb{\frac{x}{W^{1 + d \kappa/2}}}
\\
&\;=\; \sum_{x \in \Z^d} \widetilde{P}_x([\eta T]) \, \varphi \pbb{\frac{\pi(x)}{W^{1 + d 
\kappa/2}}}\,,
\end{align*}
where $\pi(x)$ is defined through $\pi(x) \in \Lambda_N$ and $x - \pi(x) \in N \Z^d$.
Define the 
sequence of i.i.d.\ random variables $A_1, A_2, \dots$ whose law is
$\frac{1}{M} \sum_{a \in \Z^d} \ind{1 \leq \abs{a} \leq W} \, \delta_a$,
where $\delta_a$ denotes the point measure at $a$.
Then we have
\begin{equation} \label{rewriting sum over paths as an expectation}
\sum_{x \in \Z^d} \widetilde{P}_x([\eta T]) \, \varphi \pbb{\frac{\pi(x)}{W^{1 + d \kappa / 2}}} 
\;=\; \E \, \varphi\pbb{\frac{\pi \pb{A_1 + \cdots A_{[\eta T]}}}{W^{1 + d \kappa/2}}}\,.
\end{equation}
Next, we introduce the partition
\begin{equation*}
1 \;=\; \indb{\absb{A_1 + \cdots + A_{[\eta T]}} < N/2} + \indb{\absb{A_1 + \cdots + A_{[\eta T]}} \geq
N/2}
\end{equation*}
in the expectation in \eqref{rewriting sum over paths as an expectation}. The second resulting 
term is bounded by
\begin{equation*}
\norm{\varphi}_\infty \, \P \pb{\absb{A_1 + \cdots + A_{[\eta T]}} \geq
N/2}\,.
\end{equation*}
This vanishes in the limit $W \to \infty$ by the central limit theorem, since $\frac{N}{W 
\sqrt{[\eta T]}} \to \infty$ by assumption.

The first term resulting from the partition is
\begin{equation*}
\E \, \varphi\pbb{\frac{A_1 + \cdots A_{[\eta T]}}{W^{1 + d \kappa/2}}} \indb{\absb{A_1 + \cdots + 
A_{[\eta T]}} < N/2} \;=\;
\E \, \varphi\pbb{\frac{A_1 + \cdots A_{[\eta T]}}{W^{1 + d \kappa/2}}} + o(1)\,,
\end{equation*}
by the same argument as above. Therefore we get
\begin{equation*}
\sum_{x \in \Lambda_N} P_x([\eta T]) \, \varphi \pbb{\frac{x}{W^{1 + d \kappa/2}}} \;=\; \E \, 
\varphi \pbb{\frac{B_1 + \cdots + B_{[\eta T]}}{\sqrt{[\eta T]}}} + o(1)\,,
\end{equation*}
where
$B_i \deq \frac{A_i}{W} \frac{\sqrt{[\eta T]}}{\sqrt{\eta}}$.
The covariance matrix of $B_i$ is $\frac{T}{d+2} \umat + o(1)$,
and the claim \eqref{application of the CLT} follows by the central limit theorem.
\end{proof}

\begin{figure}[ht!]
\begin{center}
\includegraphics[width=8cm]{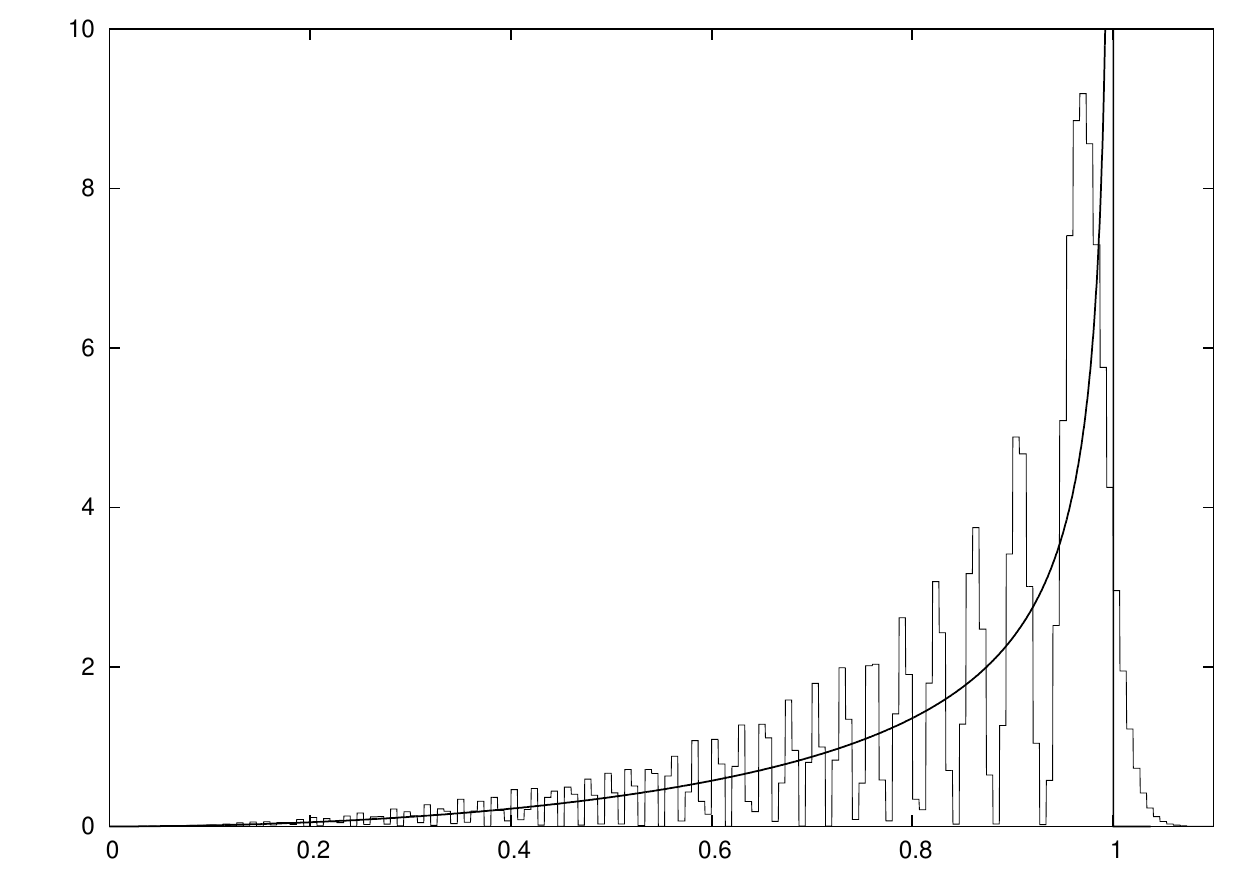}
\includegraphics[width=8cm]{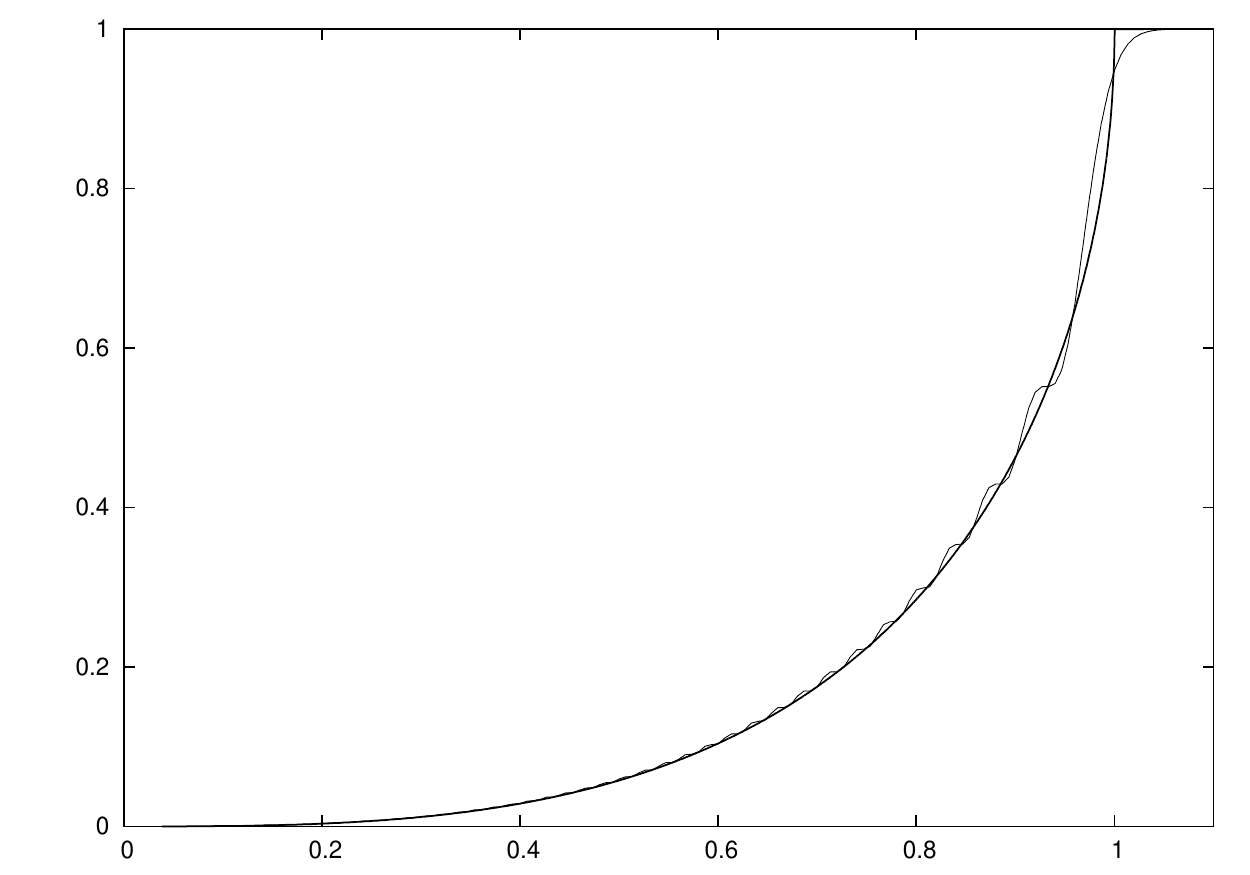}
\end{center}
\caption{The functions $f_t(\lambda)$, $f(\lambda)$ (left) and $F_t(\lambda)$, $F(\lambda)$ (right). Here we chose $t = 
150$.  \label{figure: distributions}}
\end{figure}

In a fourth and final step, we replace the probability distribution $\abs{\alpha_n(t)}^2$ with 
its asymptotic distribution. For the following we fix some test function $\varphi \in 
C_b(\R^d)$. Testing against $\varphi$ in Lemma \ref{replacing as with alphas} yields
\begin{equation} \label{ladders replaced p_n}
\sum_{n \geq 0} \abs{a_n(\eta T)}^2 \,\sum_x V_x(L_n) \, \varphi \pbb{\frac{x}{W^{1 + d \kappa 
/2}}}
\;=\; \sum_{n \geq 0} \abs{\alpha_n(\eta T)}^2 \sum_x P_x(n) \, \varphi \pbb{\frac{x}{W^{1 + \kappa 
d /2}}} + O\pbb{\frac{\norm{\varphi}_\infty}{W^{d \beta}}}\,.
\end{equation}
While the distribution $\abs{\alpha_n(t)}^2$ has no limit as $t \to \infty$, it turns out that 
the rescaled distribution,
\begin{equation*}
f_t(\lambda) \;\deq\; t \, \abs{\alpha_{[t \lambda]}(t)}^2\,,
\end{equation*}
converges weakly to
\begin{equation*}
f(\lambda) \;\deq\; \frac{4}{\pi} \frac{\lambda^2}{\sqrt{1 - \lambda^2}} \ind{0 \leq \lambda \leq 1}\,.
\end{equation*}
In order to prove this, we consider the integrated distribution
\begin{equation*}
F_t(\lambda) \;\deq\; \int_0^{\lambda} \dd \xi \; f_t(\xi)\,.
\end{equation*}
We now show that $F_t(\lambda)$ converges pointwise to $F(\lambda) = \int_0^\lambda f$. See 
Figure \ref{figure: distributions} for a graph of the functions $f_t, f, F_t, F$.

\begin{proposition} \label{proposition: asymptotics of alpha_n}
The pointwise limit
\begin{equation*}
F(\lambda) \;\deq\; \lim_{t \to \infty} F_t(\lambda)
\end{equation*}
exists for all $\lambda \geq 0$ and satisfies
\begin{subequations}
\begin{align} \label{asymptotic form of alpha 1}
F(\lambda) &\;=\; \int_0^\lambda \dd \xi \; \frac{4}{\pi} \frac{\xi^2}{\sqrt{1 - \xi^2}} \;=\; 
\frac{2}{\pi} \pB{\arcsin \lambda - \lambda \sqrt{1 - \lambda^2}} & &(\lambda \in [0,1])
\\ \label{asymptotic form of alpha 2}
F(\lambda) &\;=\; 1 & &(\lambda > 1)\,.
\end{align}
\end{subequations}
\end{proposition}
\begin{proof}
See Appendix \ref{appendix: proof of asymptotics}.
\end{proof}

In order to conclude the proof of Theorem \ref{theorem: first main result}, we need the 
following result.
\begin{proposition} \label{proposition: asymptotics of ladders}
Let $T \geq 0$. Then
\begin{equation*}
\lim_{W \to \infty} \sum_{n \geq 0} \abs{a_n(\eta T)}^2 \,\sum_x V_x(L_n) \, \varphi 
\pbb{\frac{x}{W^{1 + d \kappa /2}}}
\\
=\; \int_0^\infty \dd \lambda \; f(\lambda) \int \dd X \; G(\lambda T, X) \, \varphi(X)\,.  
\end{equation*}
\end{proposition}

Indeed, Theorem \ref{theorem: first main result} is an immediate consequence of Propositions 
\ref{proposition: main estimate on error} and \ref{proposition: asymptotics of ladders}. The 
rest of this section is devoted to the proof of Proposition \ref{proposition: asymptotics of 
ladders}.

We begin by observing that the family of probability measures defined by the densities 
$\{f_t\}_{t \geq 0}$ is tight, so that we may cut out values of $\lambda$ in the range 
$[0,\delta) \cup (1 - \delta, \infty)$.
\begin{lemma} \label{tightness of measures}
Let $\epsilon > 0$. Then there is a $\delta > 0$ and a $t_0 \geq 0$ such that
\begin{equation*}
F(\delta) + 1 - F(1 - \delta) \;\leq\; \frac{\epsilon}{\norm{\varphi}_\infty}
\end{equation*}
and
\begin{equation*}
F_t(\delta) + 1 - F_t(1 - \delta) \;\leq\; \frac{\epsilon}{\norm{\varphi}_\infty}
\end{equation*}
for all $t \geq t_0$.
\end{lemma}
\begin{proof}
By Proposition \ref{proposition: asymptotics of alpha_n} we have that
\begin{equation} \label{cutoff error}
F(\delta) + 1 - F(1 - \delta) \;\to\; 0
\end{equation}
as $\delta \to 0$. Choose $\delta > 0$ small enough that the left-hand side of \eqref{cutoff 
error} is bounded by $\frac{\epsilon}{2 \norm{\varphi}_\infty}$. Moreover, Proposition 
\ref{proposition: asymptotics of alpha_n} also implies that there is a $t_0$ such that
\begin{equation*}
F_t(\delta) + 1 - F_t(1 - \delta) \;\leq\; F(\delta) + 1 - F(1 - \delta) + \frac{\epsilon}{2 
\norm{\varphi}_\infty}
\end{equation*}
for all $t \geq t_0$.
\end{proof}

Now by \eqref{ladders replaced p_n}, Proposition \ref{proposition: asymptotics of ladders} will 
follow if we can show
\begin{equation*}
\sum_{n \geq 0} \abs{\alpha_n(\eta T)}^2 \sum_x P_x(n) \, \varphi \pbb{\frac{x}{W^{1 + d \kappa 
/2}}} \;=\; \int_0^\infty \dd \lambda \; f(\lambda) \int \dd X \; G(\lambda T, X) \, \varphi(X) 
+ o(1)\,,
\end{equation*}
i.e.
\begin{equation} \label{convergence of ladders without cutoff}
\int_0^\infty \dd \lambda \; f_{\eta T}(\lambda) \sum_x P_x([\eta T \lambda]) \, \varphi 
\pbb{\frac{x}{W^{1 + d \kappa /2}}} \;=\; \int_0^\infty \dd \lambda \; f(\lambda) \int \dd X \; 
G(\lambda T, X) \, \varphi(X) + o(1)\,.
\end{equation}
Lemma \ref{tightness of measures} implies that in order to prove \eqref{convergence of ladders 
without cutoff} it suffices to prove
\begin{equation} \label{convergence of ladders with cutoff}
\int_\delta^{1 - \delta} \dd \lambda \; f_{\eta T}(\lambda) \sum_x P_x([\eta T \lambda]) \, \varphi 
\pbb{\frac{x}{W^{1 + d \kappa /2}}} \;=\; \int_\delta^{1 - \delta} \dd \lambda \; f(\lambda) 
\int \dd X \; G(\lambda T, X) \, \varphi(X) + o(1)\,,
\end{equation}
for every $\delta > 0$.

Next, note that, by Lemma \ref{lemma: central limit argument}, the sum on the left-hand side of 
\eqref{convergence of ladders with cutoff} converges to
$\int \dd X \, G(\lambda T, X) \, \varphi(X)$
for each $\lambda \in [\delta, 1 - \delta]$. In order to invoke the dominated convergence theorem, we need an integrable 
bound on $f_t(\lambda)$.

\begin{lemma} \label{bound for dominated convergence}
Let $\delta > 0$. Then there is a $C > 0$ such that
$f_t(\lambda) \leq C$
for all $\lambda \in [\delta, 1 - \delta]$ and $t$ large enough.
\end{lemma}
\begin{proof}
From Lemma \ref{lemma: Chebyshev transform of exponential} we get
\begin{equation*}
\absb{\alpha_{[t \lambda]}(t)}^2 t \;\leq\; C \absb{J_{[t \lambda] + 1}(t)}^2 t\,.
\end{equation*}
We estimate this using the following result due to Krasikov (see \cite{krasikov}, Theorem 2).
Setting $\mu \deq (2 \nu + 1)(2 \nu + 3)$ and assuming that $\nu > -1/2$ and $t > \sqrt{\mu + 
\mu^{2/3}} / 2$, we have the bound
\begin{equation*}
\abs{J_\nu(t)}^2 \;\leq\; \frac{4}{\pi} \, \frac{4 t^2 - (2 \nu + 1) (2 \nu + 5)}{(4 t^2 - 
\mu)^{3/2} - \mu}\,.
\end{equation*}
Setting $\nu = [t \lambda] + 1$ yields
$\abs{J_{[t \lambda] + 1}(t)}^2 \leq \frac{C}{t}$
for $\lambda \in (\delta, 1 - \delta)$ and $t$ large enough. This completes the proof.
\end{proof}

By Lemmas \ref{bound for dominated convergence} and \ref{lemma: central limit argument}, it is 
enough to prove that
\begin{equation} \label{convergence of ladders with Gaussian}
\int_\delta^{1 - \delta} \dd \lambda \; f_{\eta T}(\lambda) \int \dd X \; G(\lambda T, X) \, 
\varphi(X) \;=\; \int_\delta^{1 - \delta} \dd \lambda \; f(\lambda) \int \dd X \; G(\lambda T, 
X) \, \varphi(X) + o(1)\,.
\end{equation}
Let us abbreviate
\begin{equation*}
g(\lambda) \;\deq\; \int \dd X \; G(\lambda T) \, \varphi(X)\,.
\end{equation*}
The proof of Proposition \ref{proposition: asymptotics of ladders} is therefore completed by 
the following result.
\begin{lemma}
Let $\delta > 0$. Then
\begin{equation*}
\lim_{t \to \infty}
\int_\delta^{1 - \delta} \dd \lambda \; f_t(\lambda)  g(\lambda) \;=\; \int_\delta^{1 - \delta} 
\dd \lambda \; f(\lambda)  g (\lambda) \,.
\end{equation*}
\end{lemma}
\begin{proof}
The proof is a simple integration by parts. It is easy to check that on $[\delta, 1 - \delta]$ 
the function $g$ is smooth and its derivative is bounded. We find
\begin{equation*}
\int_\delta^{1 - \delta} \dd \lambda \; f_t(\lambda)  g(\lambda)
\;=\; \int_\delta^{1 - \delta} \dd \lambda \; F'_t(\lambda)  g(\lambda)
\;=\; - \int_\delta^{1 - \delta} \dd \lambda \; F_t(\lambda) g'(\lambda) + F_t(1 - \delta) g(1 - \delta) - F_t(\delta) 
g(\delta)\,.
\end{equation*}
Proposition \ref{proposition: asymptotics of alpha_n} and dominated convergence yield the claim.
\end{proof}

\section{Symmetric matrices} \label{section: symmetric matrices}

In this section we describe how to extend the argument of Sections \ref{section: graphical 
representation} -- \ref{section: ladder} to the symmetric case \eqref{hnormalization_2}. While 
in the Hermitian case \eqref{hnormalization_1} we had
\begin{equation*}
\E H_{xy} H_{yx} \;=\; \frac{1}{M-1} \,, \qquad \E H_{xy} H_{xy} \;=\; 0\,,
\end{equation*}
we now have
\begin{equation} \label{covariance of symmetric matrix elements}
\E H_{xy} H_{yx} \;=\; \E H_{xy} H_{xy} \;=\; \frac{1}{M-1}\,.
\end{equation}
Since the distribution of $H_{xy}$ is symmetric, Lemma \ref{lemma: characterization of the set 
of partitions} also holds in the symmetric case. However, \eqref{covariance of symmetric matrix 
elements} implies that there is no restriction on the order of the labels associated with an 
edge. Thus we replace \eqref{expression for V(Gamma)} with
\begin{equation} \label{expression for V(Gamma) for symmetric matrices}
V_x(\Gamma) \;=\; \sum_{\b x} Q_x(\b x) \sum_{\b \varrho_\Gamma} \pBB{\prod_{\gamma \in \Gamma} 
\Delta_{\b x}(\varrho_\gamma)} \pBB{\prod_{\gamma \neq \gamma'} \ind{\varrho_\gamma \neq 
\varrho_{\gamma'}}} \, \frac{1}{(M-1)^{\bar n}}
\,,
\end{equation}
where
\begin{equation*}
\Delta_{\b x}(\varrho_\gamma) \;\deq\; \prod_{e \in \gamma} \ind{\varrho_{\b x}(e) = \varrho_\gamma}\,.
\end{equation*}

Next, we define the set  $\scr G^*_{n,n'}$ as the set of lumpings $\scr G_{n,n'}$ without the 
complete ladder and the complete antiladder (see its definition below).
It is easy to see that the analogue of Lemma \ref{lemma: sum over non-pairings} holds with
\begin{equation*}
R_x(\Gamma) \;\deq\; \sum_{\b x} Q_x(\b x) \sum_{\b \varrho_\Gamma} \pBB{\prod_{\gamma \in \Gamma} 
\Delta_{\b x}(\varrho_\gamma)} \, \frac{1}{(M-1)^{\bar n}}
\,.
\end{equation*}
It therefore suffices to estimate the contribution of pairings $\Gamma \in \scr P^*_{n,n'}$. We 
have that
\begin{multline} \label{direct and twisted bridges}
R_x(\Gamma) \;=\; \sum_{\b x} Q_x(\b x) \, \frac{1}{(M - 1)^{\bar n}}
\\
\times \prod_{\{e, e'\} \in \Gamma} \pB{\ind{x_{a(e)}  = x_{b(e')}} \ind{x_{b(e)} =  x_{a(e')}} 
+
\ind{x_{a(e)}  = x_{a(e')}} \ind{x_{b(e)} =  x_{b(e')}}}\,.
\end{multline}
Thus, the graphical representation of pairings has to be modified as follows. Each bridge 
$\sigma \in \Gamma$ carries a tag, \emph{straight} or \emph{twisted}, which arises from 
multiplying out the product in \eqref{direct and twisted bridges}. Twisted bridges are 
graphically represented with dashed lines.

In order to find a good notion of combinatorial complexity of pairings, we define 
\emph{antiparallel bridges} as follows. Two bridges $\{e_i, e_j\}$ and $\{e_{i+1}, e_{j+1}\}$ 
are antiparallel if $i+1, j+1 \notin \{0,n\}$; see Figure \ref{figure: antiparallel bridges}.  
An \emph{antiladder} is a sequence of bridges such that two consecutive bridges are 
antiparallel.
\begin{figure}[ht!]
\begin{center}
\includegraphics{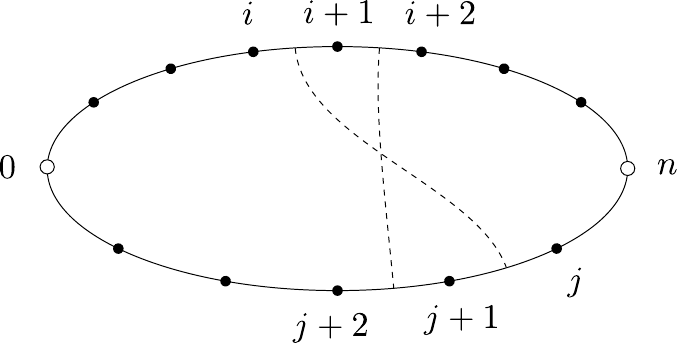}
\end{center}
\caption{Two antiparallel twisted bridges. Compare to Figure \ref{figure: parallel bridges}.
\label{figure: antiparallel bridges}}
\end{figure}
It is easy to see that, in addition to ladders whose rungs are straight bridges, antiladders 
whose rungs are twisted bridges have a leading order contribution.

The skeleton $\Sigma = S(\Gamma)$ of the pairing $\Gamma$ is obtained from $\Gamma$ by the 
following procedure. A pair of parallel straight bridges is collapsed to form a single straight 
bridge. A pair of antiparallel twisted bridges is collapsed to form a single twisted bridge.  
This is repeated until no parallel straight bridges or antiparallel twisted bridges remain.  
The resulting pairing is the skeleton $\Sigma = S(\Gamma)$; see Figure \ref{figure: fusing 
tagged bridges}.
\begin{figure}[ht!]
\begin{center}
\includegraphics{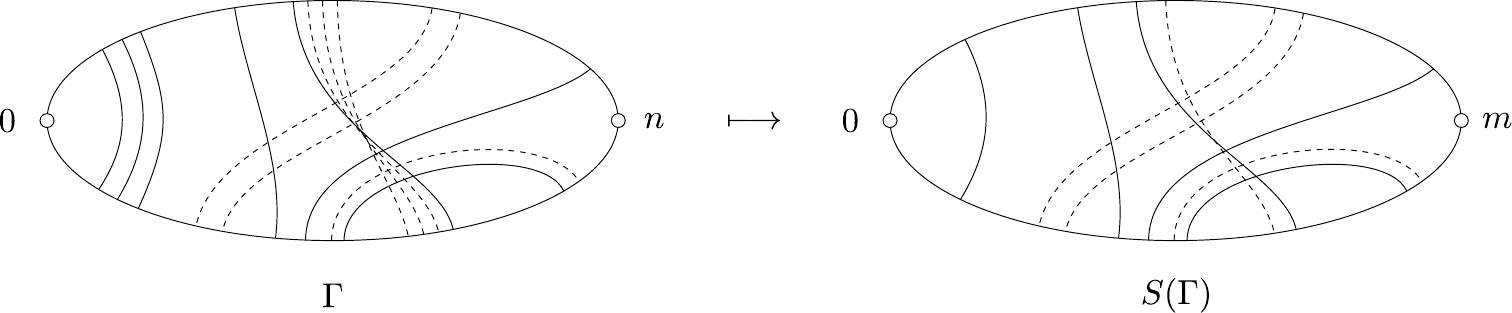}
\end{center}
\caption{The construction of the tagged skeleton graph. \label{figure: fusing tagged bridges}}
\end{figure}
Thus we see that Lemma \ref{lemma: collapsing bridges} holds. Moreover, Lemma \ref{lemma: 
properties of skeleton pairings} holds, provided that (i) is replaced with
\begin{itemize}
\item[(i')]
Each $\Sigma \in \scr S_{m,m'}^*$ contains no parallel straight bridges and no antiparallel 
twisted bridges.
\end{itemize}
Crucially, Lemma \ref{lemma: bound on number of orbits} remains valid for such tagged 
skeletons. This can be easily seen using the orbit construction of the proof of Lemma 
\ref{lemma: bound on number of orbits}, combined with (i').

Next, we associate a factor $D_{\ell}(y,z)$ with each bridge $\sigma \in \Sigma$. If $\sigma$ 
is straight, this is done exactly as in Section \ref{subsection: orbit of vertices}. If
$\sigma$ is twisted, this association follows immediately from the definition of the 
antiladder.  Thus we find that Lemma \ref{lemma: summing over orbit labels} holds. The rest of 
the analysis in Section \ref{section: error estimates} carries over almost verbatim; the only 
required modification is the summation over $2^{\bar m}$ tag configurations of the bridges of 
$\Sigma$.  The resulting factor $2^{\bar m}$ is immaterial.

Finally, the complete ladder pairing yields \eqref{main result}. The complete antiladder is 
subleading, as its contribution vanishes unless $x = 0$.

\section{Delocalization: proofs of Theorem \ref{theorem: second main result} and Corollary \ref{cor: deloc}} 
\label{section: delocalization}

In this section we show how to derive Theorem \ref{theorem: second main result} from Theorem \ref{theorem: first main 
result}, and derive Corollary \ref{cor: deloc} as a consequence.

\begin{proof}[Proof of Theorem \ref{theorem: second main result}]
We use an argument due to Chen \cite{Ch} showing that diffusive motion implies delocalization of the vast majority of 
eigenvectors.

Recall that $P_{x, \ell}(y) \deq \ind{\abs{y - x} \geq \ell}$ is the characteristic function of 
the complement (in $\Lambda_N$) of the $\ell$-neighborhood of $x$. Also, $\fra A_{\epsilon, 
\ell}^\omega$, defined by
\begin{equation*}
\fra A^\omega_{\epsilon, \ell} \;=\; \hbb{\alpha \in \fra A \,:\, \sum_x 
\abs{\psi_\alpha^\omega(x)} \, \norm{P_{x, \ell} \, \psi^\omega_\alpha} < \epsilon}\,,
\end{equation*}
is the set of eigenvectors localized on a scale $\ell$ up to an error of $\epsilon$.

By diagonalizing $H^\omega$,
\begin{equation*}
H^\omega \;=\; \sum_{\alpha \in \fra A} \lambda^\omega_\alpha \, \ket{\psi^\omega_\alpha} \bra{\psi^\omega_\alpha}\,,
\end{equation*}
we have
\begin{align*}
\normb{P_{x,\ell} \, \ee^{-\ii t H^\omega} \delta_x}^2 &\;=\; \normbb{\sum_{\alpha \in \fra A} P_{x,\ell} \, \ee^{-\ii t 
\lambda_\alpha^\omega} \, \ol{\psi_\alpha^\omega(x)} \, \psi_\alpha^\omega}^2
\\
&\;\leq\; \pbb{1 + \frac{1}{\zeta}} \normbb{\sum_{\alpha \in \fra A_{\epsilon,\ell}^\omega} P_{x,\ell} \, \ee^{-\ii t 
\lambda_\alpha^\omega} \, \ol{\psi_\alpha^\omega(x)} \, \psi_\alpha^\omega}^2
+
(1 + \zeta) \normbb{\sum_{\alpha \in \fra A \setminus \fra A_{\epsilon,\ell}^\omega} P_{x,\ell} \, \ee^{-\ii t 
\lambda_\alpha^\omega} \, \ol{\psi_\alpha^\omega(x)} \, \psi_\alpha^\omega}^2\,,
\end{align*}
for any $\zeta > 0$. Next, we observe that the norm in the first term may be bounded by 1:
\begin{equation*}
\normbb{\sum_{\alpha \in \fra A_{\epsilon,\ell}^\omega} P_{x,\ell} \, \ee^{-\ii t \lambda_\alpha^\omega} \, 
\ol{\psi_\alpha^\omega(x)} \, \psi_\alpha^\omega}^2 \;\leq\;
\normbb{\sum_{\alpha \in \fra A_{\epsilon,\ell}^\omega} \ol{\psi_\alpha^\omega(x)} \, 
\psi_\alpha^\omega}^2
\;=\; \sum_{\alpha \in \fra A^\omega_{\epsilon,\ell}} \abs{\psi^\omega_\alpha(x)}^2
\;\leq\; \sum_{\alpha \in \fra A} \abs{\psi^\omega_\alpha(x)}^2 \;=\; 1\,.
\end{equation*}
Thus we get
\begin{align*}
\normb{P_{x,\ell} \, \ee^{-\ii t H^\omega} \delta_x}^2 &\;\leq\;
\pbb{1 + \frac{1}{\zeta}} \normbb{\sum_{\alpha \in \fra A_{\epsilon,\ell}^\omega} P_{x,\ell} \, \ee^{-\ii t 
\lambda_\alpha^\omega} \, \ol{\psi_\alpha^\omega(x)} \, \psi_\alpha^\omega}
+ (1 + \zeta) \normbb{\sum_{\alpha \in \fra A \setminus \fra A_{\epsilon,\ell}^\omega} 
\ol{\psi_\alpha^\omega(x)} \, \psi_\alpha^\omega}^2
\\
&\;\leq\;
\pbb{1 + \frac{1}{\zeta}} \sum_{\alpha \in \fra A_{\epsilon,\ell}^\omega} 
\abs{\psi_\alpha^\omega(x)} \, \norm{P_{x,\ell} \psi_\alpha^\omega}
+ (1 + \zeta) \sum_{\alpha \in \fra A \setminus \fra A_{\epsilon,\ell}^\omega} 
\abs{\psi_\alpha^\omega(x)}^2\,.
\end{align*}
Averaging over $x \in \Lambda_N$ yields
\begin{align*}
\frac{1}{\abs{\fra A}} \sum_x \normb{P_{x,\ell} \, \ee^{-\ii t H^\omega} \delta_x}^2 &\;\leq\;
\pbb{1 + \frac{1}{\zeta}} \frac{1}{\abs{\fra A}} \sum_{\alpha \in \fra A_{\epsilon,\ell}^\omega} 
\sum_x \abs{\psi_\alpha^\omega(x)} \, \norm{P_{x,\ell} \psi_\alpha^\omega}
+ (1 + \zeta) \frac{1}{\abs{\fra A}} \sum_{\alpha \in \fra A \setminus \fra A_{\epsilon,\ell}^\omega} \sum_x 
\abs{\psi_\alpha^\omega(x)}^2\,.
\\
&\;\leq\;
\pbb{1 + \frac{1}{\zeta}} \epsilon + (1 + \zeta) \frac{\abs{\fra A \setminus \fra 
A^\omega_{\epsilon, \ell}}}{\abs{\fra A}}\,,
\end{align*}
by definition of $\fra A^\omega_{\epsilon, \ell}$. Therefore
\begin{equation*}
\frac{\abs{\fra A \setminus \fra A^\omega_{\epsilon, \ell}}}{\abs{\fra A}} \;\geq\;
\frac{1}{1 + \zeta} \frac{1}{\abs{\fra A}} \sum_x \normb{P_{x,\ell} \, \ee^{-\ii t H^\omega} 
\delta_x}^2 - \frac{\epsilon}{\zeta}\,.
\end{equation*}
Taking the expectation yields
\begin{equation} \label{main estimate in proof of delocalization}
\E \frac{\abs{\fra A \setminus \fra A_{\epsilon, \ell}}}{\abs{\fra A}} \;\geq\;
\frac{1}{1 + \zeta} \frac{1}{\abs{\fra A}} \E \sum_x \normb{P_{x,\ell} \, \ee^{-\ii t H} \delta_x}^2 - 
\frac{\epsilon}{\zeta}
\;=\;
\frac{1}{1 + \zeta} \, \E \normb{P_{0,\ell} \, \ee^{-\ii t H} \delta_0}^2 - 
\frac{\epsilon}{\zeta}\,,
\end{equation}
by translation invariance. Note that this estimate holds uniformly in $t$.

Next, pick a continuous function $\varphi(X)$ that is equal to $0$ if $\abs{X} \leq 1$ and $1$ 
if $\abs{X} \geq 2$. Recalling that $\varrho(t,x) = \abs{\scalar{\delta_x}{\ee^{- \ii t H/2} 
\delta_0}}^2$, we find
\begin{equation*}
\E \normb{P_{0, W^{1 + d \kappa / 2}} \, \ee^{-\ii t H/2} \delta_0}^2 \;=\; \sum_x \ind{\abs{x} 
\geq W^{1 + d \kappa/2}} \, \varrho(t, x) \;\geq\;
\sum_x \varphi \pbb{\frac{x}{W^{1 + d \kappa / 2}}} \, \varrho(t, x)\,.
\end{equation*}
Now choose an exponent $\tilde \kappa$ satisfying $\kappa < \tilde \kappa < 1/3$ and set $t = W^{d \tilde \kappa}$.  
Thus,
\begin{equation*}
\E \normb{P_{0, W^{1 + d \kappa / 2}} \, \ee^{-\ii W^{d \tilde \kappa} H/2} \delta_0}^2 
\;\geq\; \sum_x \varphi \pbb{W^{d/2 (\tilde \kappa - \kappa)} \, \frac{x}{W^{1 + \tilde \kappa 
d / 2}}} \, \varrho(W^{d \tilde \kappa}, x)\,.
\end{equation*}
Since we have
\begin{equation*}
\lim_{W \to \infty} \varphi \pb{W^{d/2 (\tilde \kappa - \kappa)} \, X} \;=\; 1
\end{equation*}
for $X \neq 0$ and $L(1,X)$ is continuous at $X = 0$, a simple limiting argument shows that 
Theorem \ref{theorem: first main result} implies
\begin{equation*}
\lim_{W \to \infty} \sum_x \varphi \pbb{W^{d/2 (\tilde \kappa - \kappa)} \, \frac{x}{W^{1 + d 
\tilde \kappa / 2}}} \, \varrho(W^{d \tilde \kappa}, x)
\;=\; \int \dd X \; L(1, X) \;=\; 1\,.
\end{equation*}
We have hence proved that
\begin{equation*}
\lim_{W \to \infty} \E \normb{P_{0, W^{1 + d \kappa / 2}} \, \ee^{-\ii W^{d \tilde \kappa} H/2} 
\delta_0}^2 \;=\; 1\,,
\end{equation*}
Plugging this into \eqref{main estimate in proof of delocalization} yields
\begin{equation*}
\liminf_{W \to \infty} \, \E \frac{\abs{\fra A \setminus \fra A_{\epsilon, W^{1 + d \kappa / 2}}}}{\abs{\fra A}} 
\;\geq\;
\frac{1}{1 + \zeta}  - \frac{\epsilon}{\zeta}\,.
\end{equation*}
Setting $\zeta = \sqrt{\epsilon}$ completes the proof.
\end{proof}

\begin{proof}[Proof of Corollary \ref{cor: deloc}]
Pick an intermediate exponent $\wt \kappa$ satisfying $\kappa < \wt \kappa < 1/3$ and abbreviate
\begin{equation*}
\ell \;\deq\; W^{1 + d \kappa /2}\,, \qquad \wt \ell \;\deq\; W^{1 + d \wt \kappa /2}\,.
\end{equation*}
Let $\alpha \in \fra B^\omega_\ell$ and let $u \in \Lambda_N$ be as in \eqref{definition of B_l}. Then we find by 
Cauchy-Schwarz
\begin{align*}
\pbb{\sum_x \abs{\psi^\omega_\alpha(x)} \, \normb{P_{x, \wt \ell} \, \psi^\omega_\alpha}}^2 &\;\leq\; \pbb{\sum_x 
\abs{\psi_\alpha^\omega(x)}^2 \exp \qbb{\frac{\abs{x - u}}{\ell}}^\gamma} \pbb{\sum_x \exp\hbb{- \qbb{\frac{\abs{x - 
u}}{\ell}}^\gamma} \normb{P_{x, \wt \ell} \, \psi_\alpha^\omega}^2}
\\
&\;\leq\; K \sum_{\abs{x - y} \geq \wt \ell} \exp\hbb{- \qbb{\frac{\abs{x - u}}{\ell}}^\gamma} \,
\abs{\psi_\alpha^\omega(y)}^2
\\
&\;\leq\; K \ee^{- \delta (\wt \ell / \ell)^\gamma}
\sum_{\abs{x - y} \geq \wt \ell} \exp\hbb{- \qbb{\frac{\abs{x - u}}{\ell}}^\gamma + \delta \qbb{\frac{\abs{x - 
y}}{\ell}}^\gamma} \,
\abs{\psi_\alpha^\omega(y)}^2\,,
\end{align*}
where $\delta > 0$ is some small constant to be chosen later. Using $(a + b)^\gamma \leq (2 a)^\gamma + (2 b)^\gamma$ we 
find
\begin{equation*}
\pbb{\sum_x \abs{\psi^\omega_\alpha(x)} \, \normb{P_{x, \wt \ell} \, \psi^\omega_\alpha}}^2 \;\leq\;
K \ee^{- \delta (\wt \ell / \ell)^\gamma}
\sum_{x,y} \exp\hbb{- \qbb{\frac{\abs{x - u}}{\ell}}^\gamma + \delta \qbb{2\frac{\abs{x - u}}{\ell}}^\gamma
+ \delta \qbb{2\frac{\abs{y - u}}{\ell}}^\gamma
} \,
\abs{\psi_\alpha^\omega(y)}^2\,.
\end{equation*}
Choosing $\delta < 2^{- \gamma}$ therefore yields
\begin{equation*}
\pbb{\sum_x \abs{\psi^\omega_\alpha(x)} \, \normb{P_{x, \wt \ell} \, \psi^\omega_\alpha}}^2 \;\leq\;
C \ell^d K^2 \ee^{- \delta (\wt \ell / \ell)^\gamma} \;\eqd\; \epsilon_W^2\,.
\end{equation*}
We have thus proved that
$
\fra B^\omega_\ell \subset \fra A^\omega_{\epsilon_W, \wt \ell}
$\,.
Then Corollary \ref{cor: deloc} follows from $\lim_{W \to \infty} \epsilon_W = 0$ and Theorem \ref{theorem: second main 
result}.
\end{proof}

\section{Critical pairings} \label{sec:further}
In this section we give an example family of pairings which are critical in the sense that they saturate the 2/3 rule 
(Lemma \ref{lemma: bound on number of orbits}). This implies that extending our results beyond time scales of order 
$W^{d/3}$ requires either a further resummation of pairings or a more refined classification of graphs in terms of their 
deviation from the 2/3 rule.

Let $k \geq 1$ and consider the skeleton pairing $\Sigma_k$ defined in Figure \ref{figure: critical skeleton graph}. It 
is a critical pairing in the sense that all orbits not containing the vertices $0,m$ consist of $3$ vertices.
\begin{figure}[ht!]
\begin{center}
\includegraphics{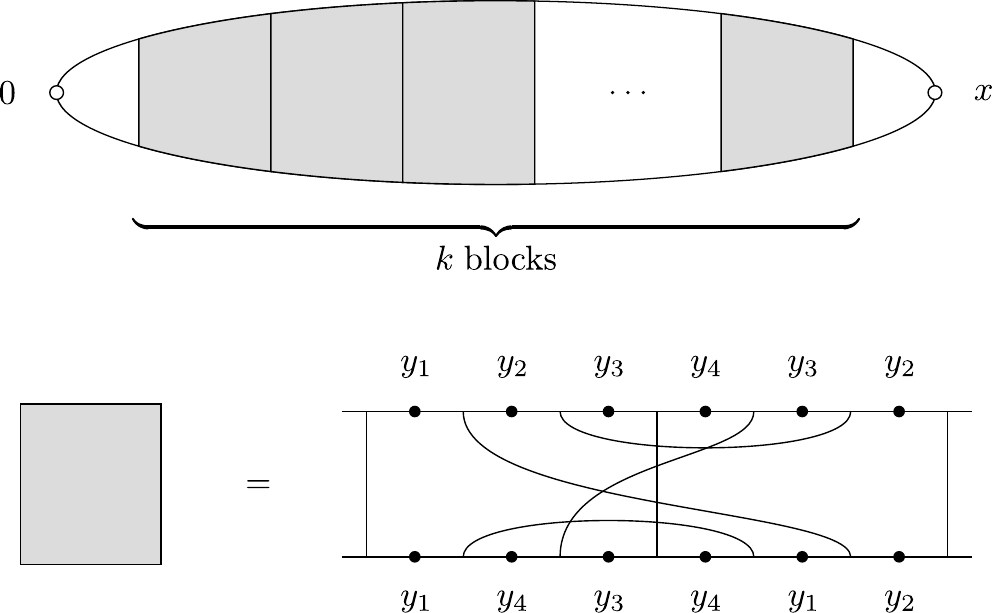}
\end{center}
\caption{A critical skeleton pairing $\Sigma_k$. The label of each vertex is indicated next to 
its vertex.  \label{figure: critical skeleton graph}}
\end{figure}

It is easy to see that for $\Sigma_k$ we have
\begin{equation*}
\bar m \;=\; 6k + 1 \,, \qquad L(\Sigma_k) \;=\; 4k + 1\,.
\end{equation*}
In particular, the 2/3 rule of Lemma \ref{lemma: bound on number of orbits} is saturated. Moreover, if $\b 
\ell_{\Sigma_k}$ satisfies $\ell_\sigma \geq 2$ for all $\sigma \in \Sigma_k$ then the associated pairing $\Gamma \deq 
G_{\b \ell_{\Sigma_k}}(\Sigma_k)$ has a nonzero contribution $V_x(\Gamma) \sim R_x(\Gamma) \approx M^{-2k}$ (here, and 
in the following, we ignore any powers of $W$ with exponent of order one). Indeed, it is easy to check that under the 
condition $\ell_\sigma \geq 2$ for all $\sigma$ the above $\Gamma$ satisfies all nonbacktracking conditions. (In fact, 
it suffices to require that $\ell_{\bar \sigma} \geq 2$, where $\bar \sigma$ is the bridge drawn as a vertical line in 
Figure \ref{figure: critical skeleton graph}.)

As shown in the Section \ref{section: error estimates} (see \eqref{bound on coefficients a_n}), the coefficients 
$a_n(t)$ essentially vanish if $n > (1 + o(1)) t$.  Setting $t = M^\kappa$ thus means restricting the summation to $n,n' 
\leq M^\kappa$.

Assume, to begin with, that we adopt the strategy of Section \ref{section: error estimates} in estimating the 
contribution of each graph, i.e.\ we use the 2/3 rule for each skeleton pairing and the $\ell^1$-$\ell^\infty$-type 
estimates from Lemma \ref{lemma: heat kernel bounds} on the edges of the associated multigraph. We show that the sum of 
the contributions of the skeleton pairings $\Sigma_k$ diverges if $\kappa > 1/3$. Indeed, noting that $n,n' \leq 
M^\kappa$ implies $\bar n \leq M^\kappa$, we find that the contribution of all $\Sigma_k$'s is
\begin{equation} \label{lower bound without heat kernel decay}
\sum_{k = 1}^{p/6} \frac{1}{M^{2k}} \sum_{\ell_1 + \dots + \ell_{6k} = p} 1\,,
\end{equation}
where $p = M^\kappa$ and the sum over $\ell_i$ is restricted to $\ell_i \geq 2$ for all $i$. Here we only sum over the 
pairing of maximal $\bar n = p$ (so as to obtain a lower bound), and set $6k+1 \approx 6k$.  Now \eqref{lower bound 
without heat kernel decay} is equal to
\begin{equation*}
\sum_{k = 1}^{p/6} \frac{1}{M^{2k}} \binom{p - 6k}{6k} \;\sim\; \sum_{k = 1}^{p/6} \pbb{\frac{C (M^\kappa - 6k)}{k 
M^{1/3}}}^{6k}\,,
\end{equation*}
which diverges as $W \to \infty$ if $\kappa > 1/3$. Hence a control of the error term at time scales $\kappa$ larger 
than $1/3$ would require further resummation of such critical pairings.

In the estimates of the preceding paragraph we did not make full use of the heat kernel decay associated with each 
skeleton bridge. For simplicity, the following discussion is restricted to $d = 1$ (it may be easily extended to higher 
dimensions; in fact some estimates are better in higher dimensions).
 Using Lemma \ref{lemma: heat kernel bounds}, we may improve \eqref{lower bound without heat kernel decay} to
\begin{equation} \label{lower bound with heat kernel decay}
\sum_{k = 1}^{p/6} \frac{1}{M^{2k}} \sum_{\ell_1 + \dots + \ell_{6k} = p} \frac{1}{\sqrt{\ell_1 \cdots \ell_{2k}}}\,;
\end{equation}
this is a simple consequence of the heat kernel bound of Lemma \ref{lemma: heat kernel bounds} and the fact that each 
six-block of $\Sigma_k$ contains two bridges in $\Sigma_k \setminus (\Sigma_k)_T$ for which we may apply the bound 
\eqref{linfty-bound2} (in which we drop the unimportant second term for simplicity). Now \eqref{lower bound with heat 
kernel decay} is bounded by
\begin{equation} \label{sum with heat kernel bounds}
\sum_{k = 1}^{p/6} \frac{1}{M^{2k}} \, p^k \binom{p}{4k} \;\sim\;
\sum_{k = 1}^{p/6} \pbb{\frac{C M^{5\kappa / 6}}{k^{2/3} M^{1/3}}}^{6k}\,,
\end{equation}
which is summable for $\kappa < 2/5$. Note, however, that the factor $k^{-6k}$ from \eqref{lower bound without heat 
kernel decay} has been replaced with the larger factor $k^{-4k}$. Recall that the factor $k^{-6k}$ is used to cancel the 
combinatorics $\bar m! \sim k^{6k}$ arising from the summation over all skeletons. In the present example this small 
factor is not needed, as the family $\{\Sigma_k\}$ is small. It is clear, however, that a systematic application of this 
approach requires a more refined classification of skeletons in terms of how much they deviate from the 2/3 rule. One 
expects that the number of skeletons saturating the 2/3 rule is small, and that they are therefore amenable to estimates 
of type \eqref{sum with heat kernel bounds}. Conversely, most of the $\bar m !$ skeletons are expected to deviate 
strongly from the 2/3 rule, so that their greater number is compensated by their small individual contributions.

Finally, we mention that the upper bound \eqref{linfty-bound}, used
in the $\ell^1$-$\ell^\infty$-type estimates above, neglects
the spatial decay of the heat kernel, i.e.\ that
\begin{equation} \label{precise form of heat kernel}
D_\ell(x,y)\sim \ell^{-1/2} e^{-(x-y)^2/\ell}
\end{equation}
for $\abs{x-y} \ll N$.
Thus a correct \emph{lower}
bound on the contribution of each skeleton graph
should have taken into account this additional decay as well.
A somewhat lenghtier calculation shows that
with the asymptotics \eqref{precise form of heat kernel} the estimate \eqref{lower bound with heat kernel decay} may
be improved to
\begin{equation*}
\sum_{k = 1}^{p/6} \frac{1}{M^{2k}} \sum_{\ell_1 + \cdots + \ell_{6k} = p} \; \prod_{j = 1}^k 
\frac{1}{\sqrt{\ell^{(j)}_2 \ell^{(j)}_4 + \ell^{(j)}_2 \ell^{(j)}_6 + \ell^{(j)}_3 \ell^{(j)}_4 + \ell^{(j)}_3 
\ell^{(j)}_6 + \ell^{(j)}_4 \ell^{(j)}_5 + \ell^{(j)}_4 \ell^{(j)}_6 + \ell^{(j)}_5 \ell^{(j)}_6}}\,,
\end{equation*}
where we abbreviated $\ell^{(j)}_i \deq \ell_{6(j - 1) + i}$. It is not hard to see that the resulting bound is the same 
as \eqref{sum with heat kernel bounds}, with a smaller constant $C$. In other words, the gain obtained from the spatial 
decay of the heat kernel is immaterial, and the $\ell^1$-$\ell^\infty$-estimates cannot be improved.

In conclusion: Our estimates rely on an indiscriminate application of the 2/3 rule to all skeleton pairings; going 
beyond time scales of order $W^{d/3}$ would require either (i) a refined classification of the skeleton pairings in 
terms of how much they deviate from the 2/3 rule, combined with a systematic use of the bound \eqref{linfty-bound2} on 
all bridges in $\Sigma \setminus \Sigma_T$; or (ii) a further resummation of graphs in order to exploit cancellations.  
The approach (i) can be expected to reach at most times of order $W^{2/5}$ for $d = 1$.

\appendix

\section{Proof of Proposition \ref{proposition: asymptotics of alpha_n}} \label{appendix: proof 
of asymptotics}
Note first that $F$ is monotone nondecreasing and satisfies $0 \leq F(\lambda) \leq 1$, as 
follows from \eqref{orthonormality of Chebyshev coefficients}. Hence it is enough to prove 
\eqref{asymptotic form of alpha 1} for $\lambda \in (0,1)$.

For the following it is convenient to replace $F_t$ with $\widetilde{F}_t$, defined by
\begin{equation*}
\widetilde F_t(\lambda) \;\deq\; \sum_{n = 0}^{[\lambda t]} \abs{\alpha_n(t)}^2 \;=\; 
F_t\pbb{\frac{[t \lambda + 1]}{t}}\,.
\end{equation*}
By Lemma \ref{bound for dominated convergence} we have $F_t(\lambda) - \widetilde F_t(\lambda) = o(1)$ as $t \to 
\infty$.

Thus let $\lambda \in (0,1)$ be fixed. From \eqref{definition of the coefficients alpha} we 
find
\begin{equation*}
\widetilde F_t(\lambda) \;=\; \sum_{n = 0}^{[\lambda t]} \absbb{\frac{2}{\pi} \int_{-1}^1 \dd 
\xi \, \sqrt{1 - \xi^2} \, U_n(\xi) \, \ee^{-\ii t \xi}}^2
\;=\; \sum_{n = 0}^{[\lambda t]} \absbb{\frac{2}{\pi} \int_0^\pi \dd \theta \; \sin \theta \sin 
[(n + 1) \theta] \, \ee^{-\ii t \cos \theta}}^2\,,
\end{equation*}
where we used \eqref{definition of Chebyshev polynomial}.  Thus,
\begin{multline} \label{expression for cumulative distribution}
\widetilde F_t(\lambda) \;=\; \frac{1}{\pi^2} \int_0^\pi \dd \theta \int_0^\pi \dd \theta' \; 
\sin \theta \, \sin \theta' \, \ee^{\ii t (\cos \theta - \cos \theta')}
\\
\times \qBB{
\frac{\ee^{\ii ([\lambda t] + 1)(\theta + \theta')} - 1}{\ee^{-\ii (\theta + \theta')} - 1}
+ \frac{\ee^{-\ii ([\lambda t] + 1)(\theta + \theta')} - 1}{\ee^{\ii (\theta + \theta')} - 1}
- \frac{\ee^{\ii ([\lambda t] + 1)(\theta - \theta')} - 1}{\ee^{-\ii (\theta - \theta')} - 1}
- \frac{\ee^{-\ii ([\lambda t] + 1)(\theta - \theta')} - 1}{\ee^{\ii (\theta - \theta')} - 1}
}\,,
\end{multline}
We now claim that the limit $t \to \infty$ of the first two terms of \eqref{expression for cumulative distribution} 
vanish by a stationary phase argument. Let us write the first term of \eqref{expression for cumulative distribution} as 
$R_t^1 + R_t^2$, where
\begin{align*}
R_t^1 &\;\deq\; \frac{1}{\pi^2} \int_0^\pi \dd \theta \int_0^\pi \dd \theta' \; \sin \theta \, 
\sin \theta' \, \ee^{\ii t (\cos \theta - \cos \theta')}
\frac{\ee^{\ii ([\lambda t] + 1)(\theta + \theta')}}{\ee^{-\ii (\theta + \theta')} - 1}
\\
&\;=\; \frac{1}{\pi^2} \int_0^\pi \dd \theta \int_0^\pi \dd \theta' \; \frac{\sin \theta \, 
\sin \theta' \, \ee^{\ii (1 - \{\lambda t\})(\theta + \theta')}}{\ee^{-\ii (\theta + \theta')} 
- 1} \, \ee^{\ii t (\cos \theta - \cos \theta' + \lambda \theta + \lambda \theta')}
\\
&\;\eqd\; \int_0^\pi \dd \theta \int_0^\pi \dd \theta' \, a_t(\theta, \theta') \, \ee^{\ii t 
\phi(\theta, \theta')}\,,
\end{align*}
where $\{\xi\} \deq \xi - [\xi] \in [0,1)$. One readily finds the bounds
\begin{equation*}
\inf_{\theta, \theta' \in [0,\pi]}\abs{\nabla \phi(\theta, \theta')} \geq \lambda\,, \qquad 
\sup_{\theta, \theta' \in [0,\pi]}\abs{\nabla^2 \phi(\theta, \theta')} < \infty\,,\qquad
\sup_t \sup_{\theta, \theta' \in [0,\pi]} \, \abs{\nabla a_t(\theta, \theta')} \;<\; \infty\,.
\end{equation*}
A standard stationary phase argument therefore yields $\lim_{t \to \infty} R_t^1 = 0$.


Similarly, we find
\begin{equation*}
R_t^2 \;=\; -\frac{1}{\pi^2} \int_0^\pi \dd \theta \int_0^\pi \dd \theta' \;  \ee^{\ii t (\cos 
\theta - \cos \theta')}
\underbrace{\frac{\sin \theta \, \sin \theta'}{\ee^{-\ii (\theta + \theta')} - 1}}_{\eqd 
b(\theta, \theta')}
\end{equation*}
As above, the functions $b$ and $\nabla b$ are bounded on $[0,\pi]^2$.
The phase $\cos \theta - \cos \theta'$ has four stationary points, $(0,0), (0,\pi), (\pi, 0), 
(\pi, \pi)$, all of them nondegenerate. Therefore a standard stationary phase argument implies 
that $R^2_t = O(t^{-1/2})$. (Note that the stationary points lie on the boundary of the 
integration domain. This is not a problem, however, as the usual stationary phase argument may 
be applied in combination with the identity $\int_0^\infty \dd x \, \ee^{\ii t x^2} = 
O(t^{-1/2})$.) Similarly, one shows that the second term of \eqref{expression for cumulative 
distribution} vanishes as $t \to \infty$.

Next, as we have just shown, we have
\begin{equation*}
\widetilde F_t(\lambda) \;=\; \widetilde F^0_t(\lambda) + \widetilde F^+_t(\lambda) + \widetilde F^-_t(\lambda) + o(1)
\end{equation*}
for $t \to \infty$, where
\begin{align*}
\widetilde F^0_t(\lambda) &\;\deq\; \frac{1}{\pi^2} \int_0^\pi \dd \theta \int_0^\pi \dd \theta' \; \sin \theta \, \sin 
\theta' \, \ee^{\ii t (\cos \theta - \cos \theta')}
\qBB{\frac{1}{\ee^{-\ii (\theta - \theta')} - 1} + \frac{1}{\ee^{\ii (\theta - \theta')} - 1}}\,,
\\
\widetilde F^\pm_t(\lambda) &\;\deq\; - \frac{1}{\pi^2} \int_0^\pi \dd \theta \; \cal P \int_0^\pi \dd \theta' \; \sin 
\theta \, \sin \theta' \, \ee^{\ii t (\cos \theta - \cos \theta')}
\frac{\ee^{\pm \ii ([\lambda t] + 1)(\theta - \theta')}}{\ee^{\mp \ii (\theta - \theta')} - 1}\,,
\end{align*}
where $\cal P$ denotes principal value.
We now show that $\widetilde F^0_t(\lambda) = o(1)$. Indeed, the expression in square brackets in the definition of 
$\widetilde F_t^0(\lambda)$ is equal to $-1$.
Exactly as above we therefore conclude that $\widetilde F_t^0(\lambda) = O(t^{-1/2})$.

Next, let us consider $\widetilde F_t^+(\lambda)$. In a first step, we replace the factor $\frac{1}{\ee^{-\ii (\theta - 
\theta')} - 1}$ with $\frac{1}{-\ii (\theta - \theta')}$. The error is
\begin{equation*}
- \frac{1}{\pi^2} \int_0^\pi \dd \theta \; \cal P \int_0^\pi \dd \theta' \; \sin \theta \, \sin 
\theta' \, \ee^{\ii t (\cos \theta - \cos \theta')}
\ee^{\ii ([\lambda t] + 1)(\theta - \theta')} \qBB{\frac{1}{\ee^{-\ii (\theta - \theta')} - 1} 
- \frac{1}{-\ii (\theta - \theta')}}\,,
\end{equation*}
which vanishes in the limit $t \to \infty$ by the above saddle point argument (the expression 
in the square brackets is an entire analytic function, and the phase $\cos \theta - \cos 
\theta' + \lambda \theta - \lambda \theta'$ has the four nondegenerate saddle points defined by 
$\sin \theta = \sin \theta' = \lambda$).

In a second step, we choose a scale $2/5 < \epsilon < 1/2$ and introduce a cutoff in 
$\abs{\theta - \theta'}$ at $t^{-\epsilon}$.
Thus we have
\begin{multline} \label{splitting of main term in asymptotic analysis}
\widetilde F^+_t(\lambda) + o(1) \;=\; \frac{1}{\pi^2} \int_0^\pi \dd \theta \; \cal P \int_0^\pi \dd \theta' \; 
\ind{\abs{\theta - \theta'} \leq t^{-\epsilon}} \sin \theta \, \sin \theta' \, \ee^{\ii t (\cos \theta - \cos \theta')}
\frac{\ee^{\ii ([\lambda t] + 1)(\theta - \theta')}}{\ii (\theta - \theta')} \,,
\\
+ \frac{1}{\pi^2} \int_0^\pi \dd \theta \int_0^\pi \dd \theta' \; \ind{\abs{\theta - \theta'} > 
t^{-\epsilon}} \sin \theta \, \sin \theta' \, \ee^{\ii t (\cos \theta - \cos \theta')}
\frac{\ee^{\ii ([\lambda t] + 1)(\theta - \theta')}}{\ii (\theta - \theta')} \,,
\end{multline}
Let us abbreviate $D_t \;\deq\; \{(\theta, \theta') \in [0, \pi]^2 \,:\, \abs{\theta - \theta' 
} > t^{-\epsilon}\}$.
The second term on the right-hand side of \eqref{splitting of main term in asymptotic analysis} 
is equal to
\begin{equation} \label{splitting of main term in asymptotic analysis, second term}
\frac{1}{\pi^2} \int_{D_t} \dd \theta \, \dd \theta' \; \ee^{\ii t (\cos \theta - \cos \theta' 
+ \lambda \theta - \lambda \theta')} \frac{\sin \theta \, \sin \theta' \, \ee^{\ii (1 - 
\{\lambda t\})(\theta - \theta')}}{\ii (\theta - \theta')}
\;\eqd\; \int_{D_t} \dd \theta \, \dd \theta' \; \ee^{\ii t \phi(\theta, \theta')} 
\frac{a_t(\theta, \theta')}{\theta - \theta'}\,.
\end{equation}
In the domain $D_t$ the phase $\phi$ has two stationary points defined by $\sin \theta = \sin 
\theta' = \lambda$ and $\theta \neq \theta'$. For all $(\theta, \theta')$ not in some fixed 
neighbourhood of these stationary points and satisfying $\abs{\theta - \theta'} > 
t^{-\epsilon}$, we have the bound
\begin{equation*}
\abs{\nabla \phi(\theta, \theta)} \;\geq\; C t^{-\epsilon}\,,
\end{equation*}
for some constant $C > 0$ depending on $\lambda$, and large enough $t$.
Thus a standard saddle point analysis shows that \eqref{splitting of main term in asymptotic 
analysis, second term} is of the order
$t^{-1/2} + t^{2 \epsilon - 1} = o(1)$.

In a third step, we analyse the first term on the right-hand side of \eqref{splitting of main 
term in asymptotic analysis}. We introduce the new coordinates
\begin{equation*}
u \;=\; \frac{\theta + \theta'}{2} \,, \qquad v \;=\; \theta - \theta'\,,
\end{equation*}
and write
\begin{align*}
\widetilde F_t^+(\lambda) + o(1) &\;=\; \frac{1}{\pi^2} \int_0^\pi \dd \theta \; \cal P \int_0^\pi \dd \theta' \; 
\ind{\abs{\theta - \theta'} \leq t^{-\epsilon}} \sin \theta \, \sin \theta' \, \ee^{\ii t (\cos \theta - \cos \theta')}
\frac{\ee^{\ii ([\lambda t] + 1)(\theta - \theta')}}{\ii (\theta - \theta')}
\\
&\;=\; \frac{1}{\pi^2} \int_0^\pi \dd u \; \cal P \int_{-a_{t,u}}^{a_{t,u}} \dd v \; \sin 
\pbb{u + \frac{v}{2}} \sin \pbb{u - \frac{v}{2}} \, \ee^{\ii (1 - \{\lambda t\}) v} 
\frac{\ee^{\ii t (\lambda v - 2 \sin u \sin \frac{v}{2})}}{\ii v}\,,
\end{align*}
where
\begin{equation*}
a_{t,u} \;\deq\; \min\{t^{-\epsilon}, 2u, 2(\pi - u)\}\,.
\end{equation*}
Now we replace the factor $\ee^{\ii t (\lambda v - 2 \sin u \sin \frac{v}{2})}$ with $\ee^{\ii 
t v (\lambda - \sin u)}$. The resulting error is
\begin{equation} \label{error controlled with power series}
R_t \;\deq\; \frac{1}{\pi^2} \int_0^\pi \dd u \; \cal P \int_{-a_{t,u}}^{a_{t,u}} \dd v \; \sin 
\pbb{u + \frac{v}{2}} \sin \pbb{u - \frac{v}{2}} \, \ee^{\ii (1 - \{\lambda t\}) v} \ee^{\ii t 
v (\lambda - \sin u)} \frac{\ee^{\ii t \sin u \, (v - 2 \sin \frac{v}{2})} - 1}{\ii v}\,.
\end{equation}
It is easy to check that, for $v \in [-a_{t,u}, a_{t,u}]$, we have
\begin{equation*}
\absbb{\frac{\ee^{\ii t \sin u \, (v - 2 \sin \frac{v}{2})} - 1}{\ii v}} \;\leq\; C t^{1 - 
\frac{5}{2} \epsilon} \frac{1}{\sqrt{\abs{v}}}\,.
\end{equation*}
Therefore
\begin{equation*}
\abs{R_t} \;\leq\; C \int_0^\pi \dd u \int_0^{2 \pi} \dd v \; t^{1 - \frac{5}{2} \epsilon} 
\frac{1}{\sqrt{\abs{v}}} \;=\; o(1)\,.
\end{equation*}
Thus we may write
\begin{equation*}
\widetilde F_t^+(\lambda) + o(1) \;=\;
\frac{1}{\pi^2} \int_0^\pi \dd u \; \underbrace{\cal P \int_{-a_{t,u}}^{a_{t,u}} \dd v \; \sin 
\pbb{u + \frac{v}{2}} \sin \pbb{u - \frac{v}{2}} \, \ee^{\ii (1 - \{\lambda t\}) v} 
\frac{\ee^{\ii t v (\lambda - \sin u)}}{\ii v}}_{\eqd I_t(u)}\,.
\end{equation*}

In a fourth step, we analyse $I_t(u)$ using contour integration.  Abbreviate $b \deq \lambda - 
\sin u$. Let us assume that $u$ satisfies $b \neq 0$.  Then, setting $z = \abs{b} t v$, we find
\begin{equation*}
I_t(u) \;=\; \cal P \int_{-\abs{b} t a_{t,u}}^{\abs{b} t a_{t,u}} \dd z \; \sin \pbb{u + 
\frac{z}{2\abs{b}t}} \sin \pbb{u - \frac{z}{2\abs{b}t}} \, \ee^{\ii (1 - \{\lambda t\}) 
\frac{z}{\abs{b}t}} \frac{\ee^{\ii z \sgn b}}{\ii z}\,.
\end{equation*}
Let us consider the case $b > 0$. Using the identity
\begin{equation*}
\cal P \frac{1}{v} \;=\; \ii \pi \delta(v) + \frac{1}{v - \ii 0}
\end{equation*}
and Cauchy's theorem, we find
\begin{equation*}
I_t(u) \;=\; \pi \sin^2(u) + \int_\gamma \dd z \; \sin \pbb{u + \frac{z}{2bt}} \sin \pbb{u - 
\frac{z}{2bt}} \, \ee^{\ii (1 - \{\lambda t\}) \frac{z}{bt}} \frac{\ee^{\ii z}}{\ii z}\,,
\end{equation*}
where $\gamma$ is the arc $\{bta_{t,u}(\cos \varphi, \sin \varphi)\,:\, \varphi \in [0, 
\pi]\}$. The absolute value of the integral is bounded by
\begin{equation*}
\int_0^\pi \dd \varphi \; \ee^{a_{t,u} \sin \varphi} \, \ee^{- bta_{t,u} \sin \varphi}\,,
\end{equation*}
which is bounded uniformly in $t$ and $b \neq 0$, and vanishes in the limit $t \to \infty$ for 
all $b \neq 0$. The case $b < 0$ is treated in the same way. In summary, we have, for each $u$ 
satisfying $\sin u \neq \lambda$, that
\begin{equation*}
\abs{I_t(u)} \;\leq\; C\,, \qquad \lim_{t \to \infty} I_t(u) \;=\; \pi \sin^2(u) \, 
\sgn(\lambda - \sin u)\,.
\end{equation*}
Hence by dominated convergence we get
\begin{equation*}
\lim_{t \to \infty} \widetilde F_t^+(\lambda) \;=\; \frac{1}{\pi} \int_0^\pi \dd u \; \sin^2(u) \, \sgn (\lambda - \sin 
u)\,.
\end{equation*}

A similar (in fact easier) analysis yields
\begin{equation*}
\lim_{t \to \infty} \widetilde F_t^-(\lambda) \;=\; \frac{1}{\pi} \int_0^\pi \dd u \; \sin^2(u) \, \sgn(\lambda + \sin 
u)\,.
\end{equation*}
Therefore we get
\begin{equation*}
\lim_{t \to \infty} \widetilde F_t(\lambda) \;=\; \frac{2}{\pi} \int_0^\pi \dd u \; \sin^2 (u) 
\,  \ind{\sin u < \lambda}
\;=\; \frac{4}{\pi} \int_0^\lambda \dd \xi \frac{\xi^2}{\sqrt{1 - \xi^2}}\,. \qedhere
\end{equation*}
This completes the proof of Proposition \ref{proposition: asymptotics of alpha_n}.


\begin{thebibliography}{AA}

\bibitem{Abr} Abrahams, E., Anderson, P.W., Licciardello, D.C., Ramakrishnan, T.V.:
{\it Scaling theory of localization: Absence of quantum diffusion in two
dimensions.} Phys.\ Rev.\ Lett.\ {\bf 42}, 673--676.

\bibitem{AZ} Anderson, G.; Zeitouni, O.: {\it  A CLT for a band matrix model.} Probab.\ Theory 
Related Fields
 {\bf 134} (2006), no.\ 2, 283--338.


\bibitem{AiMo}
Aizenman, M.\ and Molchanov, S.: {\it Localization at large disorder and at
extreme energies: an elementary derivation}, Commun.\
 Math.\ Phys.\ {\bf 157},  245--278  (1993).


\bibitem{ASW} Aizenman, M., Sims, R., Warzel, S.: {\it Absolutely continuous
spectra of quantum tree graphs with weak disorder.}
Commun.\ Math.\ Phys.\ {\bf 264} no.\ 2, 371--389 (2006).



\bibitem{And}
Anderson, P.: {\it Absences of diffusion in certain random lattices},
 Phys.\ Rev.\
{\bf 109}, 1492--1505 (1958).

\bibitem{BdeR} Bachmann, S.; De Roeck, W.: {\it From the Anderson model on a strip to the DMPK equation and random 
matrix theory}, J.\ Stat.\ Phys.\ {\bf 139}, 541--564 (2010).

\bibitem{BY} Bai, Z.D., Yin, Y.Q.: {\it Limit of the smallest
eigenvalue of a large dimensional sample covariance matrix.}
Ann.\ Probab.\ {\bf 21} (1993), no.\ 3, 1275--1294.

\bibitem{B} Bourgain, J.: {\it Random lattice Schr\"odinger operators
with decaying potential: some higher dimensional phenomena.} Lecture Notes
in Mathematics, Vol.\ 1807, 70--99 (2003).


\bibitem{Ch} Chen, T.: {\it
Localization lengths and Boltzmann limit for the
Anderson model at small disorders in dimension 3.}
J.\ Stat.\ Phys.\ {\bf 120} (2005), no.\ 1--2, 279--337.



\bibitem{D} Denisov, S.A.: {\it Absolutely continuous spectrum
of multidimensional Schr\"odinger operator.} Int.\ Math.\ Res.\ Not.\
{\bf 2004} no.\ 74, 3963--3982.




\bibitem{DPS} Disertori, M., Pinson, H., Spencer, T.:  {\it Density of
states for random band matrices.} Commun.\ Math.\ Phys.\ {\bf 232},
83--124 (2002).

\bibitem{DS} Disertori, M., Spencer, T.: {\it Anderson localization for a supersymmetric sigma model.}
Preprint arXiv:0910.3325.


\bibitem{DSZ} Disertori, M., Spencer, T., Zirnbauer, M.: {\it 
Quasi-diffusion in a 3D Supersymmetric Hyperbolic Sigma Model.}
Preprint arXiv:0901.1652.


\bibitem{Efe} Efetov, K.B.; Supersymmetry in Disorder and Chaos,
Cambridge University Press, Cambridge, 1997.

\bibitem{Elg}  Elgart, A.: {\it Lifshitz tails and localization in the three-dimensional Anderson model.}
Duke Math.\ J.\  {\bf 146}  (2009),  no.\ 2, 331--360.

\bibitem{erdosknowles}
Erd\H{o}s, L., Knowles, A.: \emph{Quantum diffusion and delocalization for band matrices with general distribution.} 
Preprint arXiv:1005.1838.

\bibitem{EPRSY}
Erd\H{o}s, L.,  P\'ech\'e, G.,  Ram\'irez, J.,  Schlein,  B., Yau, H.-T.:
{\it  Bulk universality for Wigner matrices.} Comm.\ Pure Appl.\ Math.\ {\bf 63} (2010), no.\ 7, 895--925.




\bibitem{ESY1} Erd{\H o}s, L., Salmhofer, M.,  Yau, H.-T.:
{\it Quantum diffusion of the random Schr\"odinger evolution in the scaling
  limit.} Acta  Math.\ {\bf 200}, no.\ 2, 211--277 (2008).
 




\bibitem{ESY2}  Erd{\H o}s, L., Salmhofer, M.,  Yau, H.-T.:
{\it  Quantum diffusion of the random Schr\"odinger evolution in the
scaling limit II.  The recollision diagrams.} Commun.\ Math.\ Phys.\ {\bf 271}, 1--53 (2007).

\bibitem{ESY3} Erd{\H o}s, L., Salmhofer, M.,  Yau, H.-T.:
{\it Quantum diffusion for the Anderson model in
scaling limit.} Ann.\ Inst.\ H.\ Poincare {\bf 8} no.\ 4, 621--685 (2007).

\bibitem{ESY4}  Erd{\H o}s, L., Schlein, B.,  Yau, H.-T.:
{\it Local semicircle law and complete
  delocalization  for Wigner random matrices.} Comm.\
Math.\ Phys.\ {\bf 287}, 641--655 (2009).

\bibitem{ESY5}  Erd{\H o}s, L., Schlein, B.,  Yau, H.-T.: {\it Universality of Random Matrices and Local Relaxation 
Flow.} Preprint arXiv:0907.5605.

\bibitem{ESYY} Erd{\H o}s, L., Schlein, B., Yau, H.-T., Yin, J.:
{\it The local relaxation flow approach to universality of the local
statistics for random matrices.} Preprint arXiv:0911.3687.

\bibitem{EY}  Erd{\H o}s, L.\ and Yau, H.-T.: \textit{Linear Boltzmann equation as
the weak coupling limit of the random Schr\"odinger equation.}
Commun.\ Pure Appl.\ Math.\
\textbf{LIII}, 667--735, (2000).



\bibitem{EYY} Erd{\H o}s, L., Yau, H.-T., Yin, J.: {\it Bulk universality
for generalized Wigner matrices.} Preprint arXiv:1001.3453.

\bibitem{FSo} Feldheim, O.\ and Sodin, S.: {\it A universality
result for the smallest eigenvalues of certain sample
covariance matrices.} Preprint arXiv:0812.1961.

\bibitem{FHS}  Froese, R., Hasler, D.,  Spitzer, W.:
{\it Transfer matrices, hyperbolic geometry and absolutely continuous spectrum for some 
discrete Schr\"odinger operators on graphs.}
J.\ Funct.\ Anal.\ {\bf 230} no.\ 1, 184--221 (2006).



\bibitem{FdeR}  Fr\"ohlich, J.\ and de Roeck, W.: {\it Diffusion of a massive
 quantum particle coupled to a quasi-free thermal medium in dimension $d\geq 4$.}
Preprint arXiv:0906.5178.

\bibitem{FS}
Fr\"ohlich, J.\ and  Spencer, T.:
{\it Absence of diffusion in the Anderson tight
binding model for large disorder or low energy},
Commun.\ Math.\ Phys.\ {\bf 88},
  151--184 (1983).




\bibitem{FMSS}  Fr\"ohlich, J., Martinelli, F., Scoppola, E.,
Spencer, T.: {\it Constructive proof of localization in the Anderson tight binding model.}
Commun.\ Math.\ Phys.\ {\bf 101}  no.\ 1, 21--46 (1985).

\bibitem{Fy} Fyodorov, Y.V.\ and Mirlin, A.D.: {\it Scaling properties of
localization in random band matrices: A $\sigma$-model approach.}
 Phys.\ Rev.\ Lett.\ {\bf 67} 2405--2409 (1991).

\bibitem{GradshteynRyzhik} Gradshteyn, I.S.\ and Ryzhik, I.M., {\it Tables of integrals, 
series, and products.}  Academic Press, New York, 2007.

\bibitem{gui}  Guionnet, A.:
{\it Large deviation upper bounds
and central limit theorems for band matrices.}
Ann.\ Inst.\ H.\ Poincar\'e Probab.\ Statist
{\bf 38 }, 341--384 (2002).



\bibitem{KKO} Kirsch, W., Krishna, M., Obermeit, J.: {\it Anderson
model with decaying randomness: Mobility edge.} Math.\ Z.\ {\bf 235},
421--433 (2000).


\bibitem{Kl}
Klein, A.: {\it
Absolutely continuous spectrum in the Anderson model on the Bethe
lattice.} Math.\ Res.\ Lett.\ {\bf 1}, 399--407 (1994).

\bibitem{krasikov}
Krasikov, I., {\it Uniform bounds for Bessel functions}. J.\ Appl.\ Anal.\ {\bf 12}, no.\ 1, 
83--91 (2006).

\bibitem{M} Mehta, M.L.: {\it Random Matrices.} Academic Press, New York, 1991.


\bibitem{Mi} Minami, N.:
{\it Local fluctuation of the spectrum of a multidimensional Anderson tight binding model. }
Comm.\ Math.\ Phys.\ {\bf 177}, no.\ 3, 709--725 (1996). 

\bibitem{RS} Rodnianski, I., Schlag, W.: {\it
Classical and quantum scattering for a class of long range random potentials. }
Int.\ Math.\ Res.\ Not.\ {\bf 5},  243--300 (2003).

\bibitem{Sch} Schenker, J.:  {\it Eigenvector localization for random
band matrices with power law band width.} Commun.\ Math.\ Phys.\
{\bf 290}, 1065--1097 (2009).



\bibitem{SSW}  Schlag, W., Shubin, C., Wolff, T.:
{\it Frequency concentration and location lengths for the
Anderson model at small disorders. }
J.\ Anal.\ Math.\ {\bf 88}, 173--220 (2002).

\bibitem{So1} Sodin, S.: {\it The spectral edge of some random band matrices.}
Preprint arXiv:0906.4047.


\bibitem{Spen}  Spencer, T.: {\it Lifshitz tails and localization.} Preprint (1993).



\bibitem{Spe}  Spencer, T.: {\it Random banded and sparse matrices (Chapter 23)}
to appear in ``Oxford Handbook of Random Matrix Theory'', edited by
G.\ Akemann, J.\ Baik and P.\ Di Francesco.


\bibitem{Sp1} Spohn, H.: {\it Derivation of the transport equation for
electrons moving through random impurities.}   J.\ Statist.\ Phys.\ {\bf 17}, no.\ 6.,
385--412 (1977).



\bibitem{TV} Tao, T.\ and Vu, V.: {\it Random matrices: Universality of the local eigenvalue 
statistics.}
 Preprint arXiv:0906.0510. 

\bibitem{VV} Valk\'o, B., Vir\'ag, B.: {\it Random Schr\"odinger operators on long
boxes, noise explosion and the GOE.}  Preprint arXiv:0912.0097.

\bibitem{W} Wigner, E.: {\it Characteristic vectors of bordered matrices with infinite 
dimensions.}  Ann.\ of Math.\ {\bf 62}, 548--564 (1955).

\end{thebibliography}
\end{document}